\pdfoutput=1

\documentclass[a4paper,UKenglish,cleveref, autoref, thm-restate, pdfa]{lipics-v2021}
\hideLIPIcs

\usepackage[ruled]{algorithm}
\usepackage{mathtools}

\def\corobts{CORoBTS}

\newcommand{\cO}{\mathcal{O}}
\newcommand{\T}{\mathcal{T}}
\newcommand{\D}{\mathcal{D}}
\renewcommand{\P}{\mathcal{P}}

\usepackage{newunicodechar}
\newunicodechar{ε}{\varepsilon}
\newunicodechar{Θ}{\Theta}
\newunicodechar{Ω}{\Omega}
\newunicodechar{←}{\leftarrow}
\newunicodechar{⇐}{\Leftarrow}
\newunicodechar{∀}{\forall}
\newunicodechar{∃}{\exists}
\newunicodechar{∅}{\emptyset}
\newunicodechar{∈}{\in}
\newunicodechar{∉}{\notin}
\newunicodechar{∩}{\cap}
\newunicodechar{≠}{\neq}
\newunicodechar{≤}{\le}
\newunicodechar{≥}{\ge}
\newunicodechar{⊆}{\subseteq}
\newunicodechar{⋅}{\cdot}
\newunicodechar{⌈}{\lceil}
\newunicodechar{⌉}{\rceil}
\newunicodechar{⌊}{\lfloor}
\newunicodechar{⌋}{\rfloor}

\newcommand{\image}[2]{
	\begin{figure}[ht]\centering
		\includegraphics{#1.pdf}
		\caption{#2}
		\label{img:#1}
	\end{figure}
}

\usepackage{xcolor}
\usepackage{soul} %
\usepackage{etoolbox} %
\makeatletter
\patchcmd{\SOUL@ulunderline}{\dimen@}{\SOUL@dimen}{}{}
\patchcmd{\SOUL@ulunderline}{\dimen@}{\SOUL@dimen}{}{}
\patchcmd{\SOUL@ulunderline}{\dimen@}{\SOUL@dimen}{}{}
\newdimen\SOUL@dimen
\makeatother

\definecolor{modcolor}{HTML}{b3e6ff}
\sethlcolor{modcolor}

\def\«#1\»{\def\temp{#1}\ifx\temp\empty{\color{modcolor}\hrule height 0.5em}\else\hl{#1}\fi}
\def\hlv#1
{{\color{modcolor}\hrule height #1\vskip -#1}}
\makeatother

\catcode`\
13%
\catcode`\	13%
\catcode`|13%
\catcode`#13%
\newenvironment{codelike}%
{%
\catcode`\
13%
\def
{\egroup\def
{\egroup\vbox\bgroup~}\vbox\bgroup~}%
\catcode`\	13%
\def	{{\tt~~}}%
\catcode`#13%
\def#{\color{blue}{\tt //}}%
\bgroup%
}{\egroup\vskip-1em}%
\catcode`\
5%
\catcode`\	10%
\catcode`\|12%
\catcode`#6%

\def\fnn#1{\textsf{#1}}
\def\fn#1 {\fnn{#1}}    %
\def\nil{\fnn{null}}
\def\INS{\fnn{INSERT}}
\def\RM{\fnn{REMOVE}}
\def\f#1{\textrm{#1}}   %

\newenvironment{myalg}[2]
	{\begin{algorithm}\caption{#1}\label{#2}\begin{codelike}}
	{\end{codelike}\end{algorithm}}

\title{Cache-Oblivious Representation of B-Tree Structures}

\author{Lukáš Ondráček}%
	{Faculty of Mathematics and Physics, Charles University, Malostranské náměstí 25, 118 00, Prague, Czech Republic}%
	{ondracek@ktiml.mff.cuni.cz}%
	{https://orcid.org/0000-0002-1262-0229}
	{} %
\author{Ondřej Mička}%
	{Faculty of Mathematics and Physics, Charles University, Malostranské náměstí 25, 118 00, Prague, Czech Republic}%
	{micka@ktiml.mff.cuni.cz}%
	{https://orcid.org/0000-0003-3143-4955}%
	{} %

\authorrunning{L. Ondráček and O. Mička}

\Copyright{Lukáš Ondráček and Ondřej Mička}


\relatedversiondetails{Full Version}{https://arxiv.org/abs/2209.09166}

\funding{%
	The research was supported by the Charles University, projects GA UK No.~458119 and No.~1311620
	and by SVV project number 260 575.%
}

\acknowledgements{%
	We would like to thank our supervisor Jiří Fink for workflow guidance and continuous proofreading.
	We also thank to Matej Lieskovský for his participation in the early stages of~the~research
	from which content of this paper arose and for the final paper review.%
}

\begin{document}
\maketitle

\begin{abstract}
We propose a general data structure CORoBTS
for storing B-tree-like search trees dynamically in a cache-oblivious way
combining the van Emde Boas memory layout with packed memory array.
The vEB layout is widely used for storing static such trees
or with an ad-hoc dynamization using PMA or other techniques.
We aim to simplify the design of future such data structures.

In the use of the vEB layout
mostly search complexity was considered, so far.
We show the complexity of depth-first search of a subtree and contiguous memory area
and provide better insight into the relationship between positions of vertices in tree and in memory.
We also describe how to build an arbitrary tree in vEB layout if we can simulate its depth-first search.
Similarly, we examine batch updates of packed memory array.

In CORoBTS,
the stored search tree has to satisfy that
all leaves are at the same depth
and vertices have arity between the chosen constants $a$ and $b$.
The data structure allows searching with an~optimal I/O complexity $\cO(\log_B{N})$
and is stored in linear space.
It provides operations for inserting
and removing a subtree;
both have an amortized I/O complexity
$\cO(S⋅(\log^2 N)/ B + \log_B N ⋅ \log\log S + 1)$
and amortized time complexity
$\cO(S⋅\log^2 N)$,
where~$S$ is the size of the subtree and $N$ the size of the whole stored tree.
Rebuilding an existing subtree
saves the multiplicative $\cO(\log^2 N)$ in both complexities
if the number of vertices on individual tree levels is not changed;
it is paid only for the inserted/removed vertices otherwise.

We combine the vEB memory layout with the packed memory array
the same way as in the cache-oblivious B-trees of Bender et al.~\cite{BDFC05}.
The difference in using depth-first search instead of parent pointers
would save one level of indirection in the B-trees.
Modifying cache-oblivious partially persistent array proposed by Davoodi et al. \cite{DFIO14}
to use CORoBTS improves its space complexity
from $\cO(U^{\log_2 3} + V \log U)$ to $\cO(U + V \log U)$,
where $U$ is the maximal size of~the array and $V$ is the~number of versions;
the data locality and I/O complexity of~both present and persistent reads are kept unchanged;
I/O complexity of writes is worsened by~a~polylogarithmic factor.
\end{abstract}

\newpage
\section*{Introduction}
\label{sect:intro}

The \textit{cache-oblivious memory model} designed by Frigo et al. \cite{FLPR99} is widely used
for developing and analyzing algorithms and data structures with respect to a multilevel memory hierarchy~\cite{handbook}.
In the model we have small \textit{cache} of size $M$ and possibly unlimited \textit{main memory}
both divided into \textit{blocks} of~size~$B$.
Each time, our program wants to access the main memory,
we have to load the whole block with needed data into cache unless it is already there.
We are interested in the number of those block transfers expressed asymptotically
in terms of $M$,~$B$ and input size, an \textit{I/O complexity};
accessing blocks in cache is for free.
The analyzed algorithm just uses data in the main memory without an explicit management of block transfers.
As the cache has limited size, newly loaded blocks evict blocks loaded earlier,
which are chosen according to the ideal strategy minimizing the number of transfers.
The ideal strategy is asymptotically as good as a \textit{Least Recently Used} strategy of real computers with $M$ doubled \cite{ST85,FLPR99}.
A~\textit{cache-oblivious} algorithm does not know the parameters $M$ and $B$
and so works on all levels of memory hierarchy simultaneously.
It is in contrast to \textit{cache-aware} algorithms,
which are optimized for~specific cache parameters.

Different cache-oblivious alternatives were proposed
to tree-based data structures.
For static search trees the well-known technique is to store the tree
in the \textit{van Emde Boas memory layout (vEB layout)} proposed by Prokop~\cite{Pro99},
where searching costs $\cO(\log_B N)$ memory transfers.
To obtain dynamic search trees, the vEB layout may be combined with a \textit{packed memory array (PMA)} described by~Bender et al.~\cite{BDFC05}
based on the work of Itai et al.~\cite{IKR81}.
It~stores ordered elements in a consecutive memory interleaved with~spaces of~total linear size
so that the data locality is not worsened asymptotically and
insertions and removals of elements require
only $\cO(1+(\log^2 N)/B)$ amortized memory transfers per operation.
Recently, Bender et al.~\cite{BCFCKKW22} proposed a randomized version of PMA lowering the~exponent of the logarithm to~1.5.
Other PMA alternatives were also proposed~\cite{BDF02, BH07, BFCK06, LB19}.
We keep using the well-known version of Itai in the main part of our paper
as it is simpler and more illustrative
and discuss the~alternatives in Section~\ref{sect:variants}.

In a cache-oblivious B-tree data structure~\cite{BDFC05}
all vertices of the search tree are stored in~a~PMA ordered by the vEB layout.
In other dynamic dictionaries~\cite{BDIW04}, only the data is kept in~the~PMA and a static index tree in vEB layout is built over the PMA.
Similar techniques are used in trees based on exponential structures~\cite{BCR02}.
Brodal et al. embedded the dynamic search tree into a static tree in vEB layout~\cite{BFJ02}.
Umar et al.~\cite{UAH13} proposed a~more complex search tree data structure utilizing the vEB layout,
which adds support for concurrency.
All~of~these data structures come up with their own balancing strategies,
which cannot be easily replaced to utilize the tree structure for other purposes
such as storing multidimensional keys without total ordering.
In some cases, indirection is used eliminating the tree structure completely~\cite{BDFC05,BDIW04,BFJ02}.
Saikkonen and Soisalon-Soininen created a dynamic memory layout for any binary search tree using rotations~\cite{SSS16},
which is however only cache-aware.

\bigskip

The cited papers analyze only the search complexity of the vEB layout
and we are not aware of papers with more in depth analysis.
We provide better insight into the relationship between positions of vertices in the tree and in memory
in the vEB layout
by describing how a contiguous part of memory looks in the tree
and how a subtree is divided into several contiguous parts of memory.
We show that the complexity of depth-first search of a subtree is the same as of sequential scan, which is optimal;
DFS of a contiguous memory area is higher just by a constant number of branch traversals.
We provide a simple algorithm for navigating during the DFS to visit just
vertices in the given contiguous memory and their ancestors.
We also describe how to build an arbitrary tree in vEB layout using DFS if we can simulate DFS of the resulting tree;
the complexity is the same as of DFS traversal of the existing tree.
Similarly, we examine batch updates of packed memory array.

We propose a data structure for
Cache-Oblivious Representation of B-Tree Structures, abbreviated as \corobts{}.
It provides direct access to the tree
and allows managing its~shape
by~inserting, removing and rebuilding subtrees,
provided that these operations preserve the standard B-tree invariants;
i.e., arity of internal vertices is bounded by the chosen integer constants $a<b$
and all leaves are at the same depth.

Searching in the tree costs $\cO(\log_B N)$ memory transfers, which is I/O optimal.
Inserting or removing a subtree of size $S$ has amortized I/O complexity
$\cO(S⋅(\log^2 N)/ B + \log_B N ⋅ \log\log S + 1)$.
Rebuilding a subtree avoids the multiplicative $\cO(\log^2 N)$ as long as
the number of vertices on individual levels in the subtree is not changed;
in general, the complexity is $\cO(S/B + K⋅(\log^2 N)/B + L⋅\log_B N + 1)$,
where $L$ is the number of levels with changed number of vertices
and $K$ the sum of changes of number of vertices on that levels.
The space usage is linear in the number of vertices.

The idea behind the data structure is similar to that of the cache-oblivious B-trees of~Bender~et~al.~\cite{BDFC05}:
we store vertices ordered by the vEB layout
interleaved with spaces in a PMA.
As modifications involve movements of existing vertices in memory
we need to recalculate pointers between them.
In contrast to Bender's B-trees, we do not store pointers to~parents in the tree;
this way, pointer updates have lower I/O complexity.
In particular, updating pointers leading to the contiguous memory of size $K$ originally costs $O(K)$ memory transfers,
because vertices under that area in the tree may be spread in different blocks;
our depth-first search visiting all ancestors but no descendants of the area costs only $O(K/B + \log_B N + 1)$ memory transfers.
The final complexities of Bender's B-trees and \corobts{} cannot be compared as the interface is totally different,
but if we use \corobts{} inside Bender's B-trees,
we can reach the same complexity using just one level of indirection instead of two.

Since we defined \corobts{} as a standalone data structure with general interface
for accessing and modifying the tree
it may avoid the need of ad-hoc constructions combining vEB and PMA in the future.
Using it as a black box will subsequently simplify description of derived data structures
and possible library implementation also putting them into practice.

We illustrate the future simplification of descriptions on one more detail of the Bender's B-trees.
They need to perform standard B-tree splits and merges of vertices
and they precisely describe how to reorder vertices in memory during that operation according to the strict vEB ordering.
As part of \corobts{}, we describe the general rebuild procedure,
which can be directly used while keeping the complex part hidden inside.
We only have to describe how to simulate DFS of the tree with one vertex split (or two siblings merged)
with access to the original tree, which is straightforward, and we get asymptotically the same complexity.

\bigskip

Our original motivation for the \corobts{} data structure
is to use it as part of the cache-oblivious partially persistent array
proposed by Davoodi et al. \cite{DFIO14} to improve its space complexity.
The persistent array can be used like an~ordinary array,
which in addition stores all its previous versions and allows their reading.
We can write only to the last (\textit{present}) version,
which causes creation of~a~new version.
By~storing another, \textit{ephemeral}, cache-oblivious data structure inside the persistent array,
it becomes persistent too
and~it can still benefit from its cache-oblivious design
as~the~persistent array tries to~preserve data locality.
We note that persistent arrays were studied earlier by Dietz~\cite{Die89},
but not in the context of cache-oblivious model;
we are not aware of any other cache-oblivious persistent array than the above-mentioned one.

Internally, the persistent array of Davoodi et al. uses (2,3)-ary search trees,
where each internal vertex represents a subarray,
its two original children further partition the subarray
and its optional third child stores the newer version of one of its siblings.
Each such tree contains at most $\cO(U\log U)$ vertices,
where $U$ is the size of the present array.
In~spite of~it, the newest such tree always occupies much larger space~$Θ(U^{\log_2 3})$.
If we store the trees in~\corobts{},
they will occupy just the~linear space in number of their vertices
and the total space usage of~the~persistent array
decreases from $\cO(U^{\log_2 3} + V\log U)$
to $\cO(U + V\log U)$, where $V$ is the~number of versions.

The I/O complexities of the persistent array are analyzed
w.r.t.~a given \textit{ephemeral} algorithm running in the persistent memory.
It executes a sequence of reads or writes that would require $T(M',B')$ block transfers when executed ephemerally
on an initially empty cache with smaller size $M'$ and smaller block size $B'$.

The original I/O complexities of present read, persistent read and write are
$\cO(T(M/3, B))$,
\[
	\cO\left(T\left(\frac{\frac{1}{3}M}{B^{1-\frac{1-ε}{\log 3}}(1+\log_B U)}, B^\frac{1-ε}{\log 3}\right)⌈\log_B U⌉\right),
	\cO\left(T\left(\frac{\frac{1}{3}M}{B^{1-\frac{1-ε}{\log 3}}\log B}, \frac{B^\frac{1-ε}{\log 3}}{\log B}\right)\log_B U\right),
\]
respectively,
as stated in \cite{DFIO14},
where $0 < ε < 1$ is a parameter of the data structure
which can be set arbitrarily.

In the persistent array using the \corobts{} data structure,
the I/O complexity of reads is kept unchanged as well as the data locality;
the I/O complexity of writes is
\[
	\cO\left(\frac{1}{ε}⋅T\left(\frac{εcM}{B^{1-\frac{1-ε}{\log 3}}\log U}, \frac{B^\frac{1-ε}{\log 3}}{\log B}\right)\log_B U⋅\log^2 U\right)
\]
for $0<ε≤1/2$ and a sufficiently small constant $c$.
The multiplicative $\log^2 U$ term is the overhead of PMA
and the memory of the internal algorithm had to be lowered by a factor of $(\log_B U)/ε$ to avoid some external pointers leading inside the tree managed by CORoBTS.
We note that the constant $c$ is a little lower than in the original paper
and the leading constant $1/ε$ was also not present there.

\section{Overview}
\label{sect:overview}

\paragraph*{Terminology}
The trees used in this paper are rooted and the order of children is significant.
We define the \textit{depth}~$d(v)$ of a vertex $v$ in a tree as the length of the path from the root to the vertex;
the~depth of the root is zero.
The \textit{level} $k$ of a tree is the set of all its vertices with depth $k$.
The \textit{height} $h(T)$ of~a~tree $T$ is the~number of its non-empty levels.
A \textit{branch} of a tree is a~path from its root to~a~leaf.

We need to distinguish three different orderings of vertices in trees.
We say that a vertex is to~the~\textit{left/right} of another one to refer to the ordering on the same level in the tree.
We use \textit{up/down} ordering on a branch.
When speaking about vertices stored in memory, we use \textit{before/after} to~refer to their memory location.
We also use \textit{memory-first} and \textit{memory-last} to speak about the~first or last vertex in a given area of memory.
Furthermore, we adopt a left/right ordering of~branches from their leaves, when they are on the same level,
and a~vertex is to~the~left/right of a branch
if it is to the left/right of the vertex on the same level of that branch.

\paragraph*{\corobts{} Interface and Invariants}
The \corobts{} data structure maintains a search tree $\T$ satisfying the following invariants:
\begin{description}
	\item[Degree invariant:] Every internal vertex has between $a$ and $b$ children
		for constants $2 ≤ a <b$.
	\item[Height invariant:] All leaves are on the same level.
\end{description}

We denote by $N = |\T|$ the number of vertices of $\T$.
The height of the tree, denoted by $H = h(\T)$, is fixed and can be changed only by a~complete rebuild.
The data stored in vertices can be only of constant size
and the tree has to be \textit{searchable},
which means that we can navigate from the root to a given vertex using only the data
stored in the~vertices on~the~path and $\cO(1)$ of other data.

The data structure provides a \textit{direct access} to the tree,
so user can traverse it and perform searches as usual
using pointers to children in vertices;
pointers to parents are not stored.
Along with searching, we allow updates by inserting and removing whole subtrees
and possibly rebuilding existing subtrees.
\textit{Inserting a subtree} as the $c$-th child of a vertex $v$
in its basic form means
creating an $a$-ary subtree under $v$ of height forced by the height invariant.
The input of this operation is just $v$ and $c$.
After the subtree creation,
the user fills the vertices with their data -- so far only children pointers have been set.
In the general version, we can insert an arbitrary subtree
provided that we can simulate its left-to-right depth-first search
and the invariants are not violated.
\textit{Removing a subtree} rooted at the $c$-th child of a vertex $v$ just destroys that subtree.
\textit{Rebuilding a subtree}
allows arbitrary changes in its structure according to the simulated depth-first search,
but more efficiently than by removing the subtree and inserting another one;
we may also rebuild several siblings subtrees into different number of subtrees without touching others.

We store vertices ordered in memory by the van Emde Boas memory layout
to allow efficient traversal of the tree in terms of I/O complexity.
We use packed-memory array to maintain holes between vertices in memory
in order to allow efficient insertions and removals of vertices in their fixed position in the sequence.

\paragraph*{Van Emde Boas Memory Layout}
The van Emde Boas memory layout was originally proposed by Prokop~\cite{Pro99}.
We use a version parametrized by a constant $0 < ε ≤ 1/2$.
It guarantees
that a branch of the tree $\T$ is contained within $\cO((\log_B |\T|)/ε + 1)$ cache blocks,
and so searching in the tree has this I/O complexity.
The~time complexity stays linear in the length of the traversed path.
We formally define the layout and a tree decomposition $\D$ based on it
in Section~\ref{sect:veb}
and prove there the search I/O complexity and other properties of the layout.

We note that the configurable parameter $ε$ is required by the persistent array,
where it gets into an exponent of some terms of the I/O complexity.
In \corobts{} itself, the value of $1/2$ is optimal.
We keep $ε$ in asymptotic complexity even as the multiplicative constant
to better illustrate its impact there
since it can be arbitrarily chosen without restricting the interface of the data structure.
On the other hand, we do not use $a$ and $b$ multiplicatively,
as they are forced by the specific application.

In Section~\ref{sect:veb_theory},
we describe the relationship between positions of vertices in the memory and in the tree.
We show how a contiguous memory interval looks in the tree,
e.g. that it may be divided into at most two specific distant parts in the tree,
and that a subtree of size $S$ may be divided into at most $\cO(\log\log S)$ memory intervals,
which can be easily located by just two branch traversals.

In Section~\ref{sect:veb_dfs},
we proof that a left-to-right depth-first-search of a subtree
has the optimal complexity $\cO(S/B + ε^2)$ if we begin at the subtree root.
The complexity of DFS of all vertices in a memory interval and their ancestors is by the search complexity higher
as there may be possibly longer paths connecting root with the two separate subtrees containing the memory interval;
it can be expressed as $\cO(S/B + (\log_B N)/ε + 1)$.
We describe there an algorithm for navigating during the DFS of a memory interval to visit just the necessary vertices.

In Section~\ref{sect:veb_build}, we show how to use DFS to build a fixed-arity subtree in a reserved memory.
We provide the algorithm for building the general subtree according to its simulated DFS in Section~\ref{sect:veb_general_build}.
The complexity of subtree build is in both cases the same as the complexity of subtree DFS.

\paragraph*{Packed-Memory Array}
The \textit{packed-memory array} (PMA) is
based on~a~work of~Itai et al.~\cite{IKR81}
as described by Bender, Demain and Farach-Colton \cite{BDFC05} in the context of cache-oblivious model.
PMA stores a sequence of $N$ elements in an~array of~size~$cN$ for~$c > 1$, subject to operations
insert and delete.
Both of them have amortized I/O complexity $\cO((\log^2 N)/B + 1)$
and amortized time complexity $\cO(\log^2 N)$ if we have
a~pointer to the predecessor (insert) or the~element (delete).
Moreover, insert or delete of~$S$~consecutive elements have amortized I/O complexity $\cO(S⋅(\log^2 N)/B + 1)$.
The important property of~PMA is that any $S$ consecutive elements are stored within $\cO(S)$ consecutive cells of~the~array,
and so scanning elements in~a~PMA has the same asymptotic complexity as if the~elements were stored in~an~ordinary array.
We describe the original PMA in Section~\ref{sect:original_pma}.

Inserting or removing an item causes PMA to move all items within a consecutive part of the memory, \textit{memory interval}.
As we are storing a tree inside PMA, this would normally break pointers connecting its vertices.
To avoid it, we modify PMA so that it first identifies the memory interval and computes where each vertex will be moved,
then calls our procedure for updating pointers leading to these vertices,
and finally performs the movements.

We also need to minimize the number of calls to our pointer recalculation procedure as it brings additive terms to the I/O complexity.
For this reason, we need our PMA to process all the insert/remove operations as one batch update,
so that multiple insertions/removals at one place (or close enough) result in just one memory interval for pointer recalculation.
We describe our version of PMA in Section~\ref{sect:our_pma}.

We call $\P$ our instance of PMA
and we use it also for referring to the~underlying array since the array is the only memory used by the PMA.

\paragraph*{Pointer Recalculation}

To update the pointers leading to the memory interval $P$ given by PMA
from the parents of the contained vertices,
we use DFS mentioned above, which has the complexity $\cO(|P|/B + (\log_B N)/ε + 1)$.
It is almost as good as the sequential scan of the memory interval, which has the I/O complexity of $\cO(|P|/B)$.
The sequential scan occurs during PMA update, so we can hide this term in the amortized complexity.
We note that avoiding parent pointers allowed us to skip visiting children of vertices in $P$,
which might be scattered in up to $Θ(|P|)$ cache blocks.

The amortized I/O complexity of modified PMA is $\cO(S⋅(\log^2 N)/B + L⋅(\log_B N)/ε + 1)$,
where $S$ is the total number of vertices to insert/remove
and $L$ is the number of consecutive places where the insertions/removals occur
and so the upper bound on the number of calls of DFS.
We present the proof in Section~\ref{sect:corobts_recalc}.

\paragraph*{Subtree Insertion and Removal}

The complexity of inserting or removing a subtree is
$\cO(S⋅(\log^2 N)/B + (\log_B N ⋅ \log\log S)/ε + 1)$,
as the main part is the PMA complexity
and the subtree may be divided into at most $\cO(\log\log S)$ memory intervals;
the fixed-arity or general build is also bounded by the complexity.
See Section~\ref{sect:corobts_insert_remove}.

\paragraph*{Subtree Rebuild}

To rebuild a subtree of size $S$,
we use the fact that its vertices can be arbitrarily permuted in the memory without getting PMA involved,
which saves the $\cO(\log^2 N)$ factor.
The memory of the subtree may however be partitioned into multiple memory intervals, each containing several levels of the subtree;
changing the number of vertices in them costs $\cO(K⋅(\log^2 N)/B + L⋅(\log_B N)/ε + 1)$,
where $K$ is the sum of changes of number of vertices in each memory interval
and $L≤K$ is the number of memory intervals where inserts/removals occurred;
we can upperbound the complexity by counting each subtree level as an individual memory interval.
See Section~\ref{sect:corobts_rebuild} for details.

\paragraph*{Initialization}

We can initialize the data structure either as an $a$-ary tree
or using the general build procedure from simulated DFS.
In both cases,
we can use the worst-case PMA complexity $\cO(N/B + 1)$ as its potential will be zero after the build.
The whole initialization is then bounded by the complexity of the build, i.e. $\cO(N/B + ε^2)$.

\bigskip
We provide the \corobts{} data structure with operations of inserting an $a$-ary subtree and removing subtree
in a form of pseudocodes in Appendix~\ref{sect:appendix_algorithms}
and discuss different variants of \corobts{} in Section~\ref{sect:variants}.
Our application in persistent array is described in Section~\ref{sect:persistent_array}.

\section{Van Emde Boas Memory Layout}
\label{sect:veb}
We store vertices in memory ordered by the so called van Emde Boas memory layout~\cite{Pro99}.
The~main idea is that we cut the tree in a fraction of its height
and recursively store its upper subtree followed by all its bottom subtrees.
We require that the upper subtree is always lower than or of same height as the bottom subtrees
and neither of them is empty.
Formally, we use the~following definition, parametrized by $ε$:

\begin{definition}[van Emde Boas permutation]\label{def:van_emde_boas}
	For a tree $T$ with all leaves at the same depth and $0<\varepsilon≤1/2$,
	we recursively define $\f{vEB}_{\varepsilon}(T)$ to be a permutation of vertices of $T$ s.t.:
	Unless height $h$ of $T$ is 1 (and thus the permutation is unique),
	we cut all edges between levels $\max(\lfloor\varepsilon h\rfloor,1) -1$ and $\max(\lfloor\varepsilon h\rfloor,1)$
	and denote the resulting subtrees $R_0, \dots, R_k$,
	where $R_0$ is the \textit{top tree} containing the root of $T$
	and $R_1, \dots, R_k$ are the remaining \textit{bottom trees} ordered left-to-right.
	Then $\f{vEB}_\varepsilon(T) := \f{vEB}_\varepsilon(R_0), \f{vEB}_\varepsilon(R_1), \dots, \f{vEB}_\varepsilon(R_k)$.
\end{definition}

We further define \textit{decomposition} $\D$ of $\T$ as the tree of the recursion in the $\f{vEB}_\varepsilon(\T)$ definition
identifying vertices of $\D$ with the corresponding subtrees of $\T$,
which we call \textit{decomposition subtrees}.
The root of $\D$ is then $\T$ and the~leaves are all one-vertex subtrees of $\T$ ordered as in memory.
See Figure~\ref{img:figure1}.

We usually denote vertices of $\T$ by $u,v,w$ and
we denote subtrees of $\T$ and vertices of $\D$ by $T,R,S$,
possibly combined with subscripts.
We note that it is important to differentiate between height as a tree property and depth as a vertex property.

\image{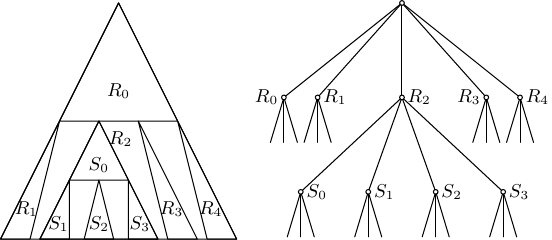}{
	The tree $\T$ on the left and its decomposition $\D$ on the right.
}

\begin{lemma}[search I/O complexity]\label{lem:search_complexity}
	Traversing a contiguous part of a branch of length $s$ has I/O complexity $\cO(s/(ε⋅\log B) + 1)$.
	Traversing a branch of a decomposition subtree $T$ has I/O complexity $\cO((\log_B |T|)/ε + 1)$.
\end{lemma}
\begin{proof}
	Let $k$ be maximal such that any subtree of height at most $k$ contains at most $B$ vertices.
	Since our tree is at most $b$-ary, we have $k ∈ Ω(\log_b B)$.
	Let us assign a tree $T_v∈\D$ to each vertex $v$ of the traversed path
	so that the height of $T_v$ is highest possible but at~most~$k$.
	We split the traversed path into segments each sharing the same tree $T_v$.

	The length of each segment except for possibly shorter first and last
	is the same as height of~its subtree $T_v$,
	which is at least $\max(\lfloor(k+1)⋅\varepsilon\rfloor,1)$.
	The number of segments is thus at~most
	$\cO(s/εk+1) ⊆ \cO(s/(ε⋅\log_b B) + 1)$.

	Finally, we prove that traversal of each segment requires at most $\cO(1)$ memory transfers:
	By choice of $k$, $|T_v| ≤ B$ for each $T_v$.
	From the definition of the decomposition, each $T_v$ is stored contiguously in memory (except for PMA holes), occupying at
	most $\cO(|T_v|)$ memory cells.
	Traversing a segment then requires only $\cO(|T_v|/B + 1) ⊆ \cO(1)$ memory transfers. 

	This gives us the total I/O complexity $\cO(s/(ε⋅\log_b B) + 1)$, which can be upperbounded
	by~$\cO((\log_B |T|)/ε + 1)$ for a branch of a decomposition subtree $T$, since the height of~the~subtree is $\cO(\log_a |T|)$.
\end{proof}

\subsection{Position of vertices in memory and in tree}
\label{sect:veb_theory}

Here, we present several properties of vEB permutation,
which we later use to find all vertices of a memory interval in tree during depth-first-search,
to find memory intervals occupied by a subtree
and to prove complexities of our algorithms.

\begin{observation}[order of vertices in memory]\label{obs:memory_order}~
	\begin{enumerate}[(a)]
		\item Vertices on a branch are stored ordered from root to leaf.
		\item Vertices of one level in $\T$ are stored ordered from left to right.
	\end{enumerate}
\end{observation}

\subsubsection{Memory interval in tree}

We partition a set of vertices by branches as follows:
\begin{definition}[tree partitioning]\label{def:tree_partitioning}
	Let $B_1, \dots, B_k$ be branches of a tree $T$ ordered left-to-right and $U$ a subset of vertices of $T$, s.t. $∀i\;B_i ∩ U = ∅$.
	For $i ∈ \{0, \dots, k\}$ let $W_i$ be the set of~all vertices of $T$
	which are to the right of $B_i$ (if $i>0$) and to the left of $B_{i+1}$ (if $i<k$),
	let~$U_i = U ∩ W_i$.
	If all $U_i$ are non-empty,
	we say that $U$ \textit{can be partitioned} by branches $B_1, \dots, B_k$
	and $\{U_i\mid i∈\{0, \dots, k\}\}$ is its \textit{partitioning}.
\end{definition}

\begin{lemma}[memory interval in tree]\label{lem:interval_in_tree}
	Given a set A of all vertices stored in a contiguous part of memory, the following holds:
	\begin{enumerate}[(a)]
		\item For each branch $B$: $A ∩ B$ induces a (possibly empty) subpath.
		\item For each level $L$ of vertices in tree order, $A ∩ L$ is a contiguous subsequence of L.
		\item Considering sequence $B_1, \dots, B_k$ of all branches intersecting $A$ in left-to-right order,
			then the sequence of depths of bottommost vertices of these branches contained in A is non-increasing,
			formally $(\max_{v ∈ B_i ∩ A} d(v))_{i=1}^k$ is non-increasing.
		\item Similarly to (c), the sequence of depths of topmost vertices, i.e. $(\min_{v ∈ B_i ∩ A} d(v))_{i=1}^k$, is non-increasing.
		\item If $A$ can be partitioned into two parts by a non-intersecting branch,
			then all vertices from the right component are in memory before all vertices from the left component.
		\item A can be partitioned into at most two components by any set of branches.
	\end{enumerate}

	Given a branch $B$ and disjoint memory intervals $A_1, \dots, A_k$ not intersected but partitioned by~$B$, then:
	\begin{enumerate}[(g)]
		\item Each $A_i$ lies in memory between a vertex of $B$ and its parent and
			for each vertex of $B$ and its parent, at most one $A_i$ lies between them in memory.
	\end{enumerate}
\end{lemma}
\begin{proof}~
	\begin{enumerate}[(a)]
		\item
			By Observation~\ref{obs:memory_order}~(a), the order of vertices on branch and in memory is the same,
			therefore memory interval induces branch subpath.
		\item
			By Observation~\ref{obs:memory_order}~(b), the order of $L$ in tree and in memory is the same.
		\item Let $u,v ∈ A$ and $w$ be vertices s.t. $u$ and $w$ are on the same branch, $d(u) < d(w) = d(v)$ and $v$ is to the right of $w$ in tree;
			then by Observation~\ref{obs:memory_order} $w$ is between $u$ and $v$ in memory, thus $w ∈ A$.
		\item Analogous to (c).
		\item
			We prove that for any pair of memory consecutive vertices $u$, $v$ (in this order) s.t.
			each of them is in different component, $u$ is to the right of $B$ and $v$ to the left of $B$.
			It thus follows that there can be only one such memory transition
			and so each vertex before $u$ is in the right component
			and each vertex after $v$ in the left one.

			Let $u$, $v$ be the pair of vertices as defined above.
			The two one-vertex subtrees of vertices $u$ and $v$ on the bottommost level of $\D$ are next to each other,
			since leaves in $\D$ are ordered as in the memory.
			Let us call their sibling ancestors (children of the lowest common ancestor) $T_u$, $T_v$.
			Both $T_u$ and $T_v$ are stored contiguously in memory and
			the~order of~$u$,~$v$ implies that $T_u$ is just before $T_v$.
			The vertex $u$ is at the end of $T_u$ and $v$ at the beginning of $T_v$;
			thus $u$ is the rightmost leaf of $T_u$ and v is the root of $T_v$.

			From the decomposition it follows that
			the roots of subtrees $T_u$ and $T_v$ either are on the~same level in $\T$,
			or the leftmost leaf of $T_u$ is parent of $T_v$.
			See the two cases in~Figure~\ref{img:figure2} without extension.

			In the former case,
			$u$ is in the left component while $v$ is in the right one.
			The rightmost branch containing $u$ is, however,
			just to the left of the leftmost branch containing $v$;
			no separating branch is therefore present and
			this case cannot happen.
			In the latter case,
			$v$~is in the left component and $u$ in the right one and
			separating branch B crosses~$T_u$ between $v$ and $u$.
			This is the only possible case.
		\item
			Keeping our partitioning of $A$ from (e) as well as definitions of $u$, $v$, $T_u$, $T_v$ and $B$,
			let us assume for contradiction that there exists another branch $B'$ partitioning W.L.O.G. the right component.
			We now have $v$ in the left component, $u$ in the middle component,
			and there exists another pair of memory-consecutive vertices $u'$, $v'$ s.t.
			$u'$ is in the right component and $v'$ in the middle component (possibly $u = v'$).
			The vertex $u'$ is separated from $u$ by $B'$, and so it cannot be in $T_u$,
			because $u$ lies on the rightmost branch of $T_u$ (see extended second case of Figure~\ref{img:figure2}),
			but still $u'$ is before $u$ in memory and so it is before the whole $T_u$;
			therefore the whole $T_u$ is contained in $A$,
			which contradicts with $B ∩ A = ∅$.
		\item
			Let $w ∈ B$ and $w' ∈ B$ parent of $w$.
			By Observation~\ref{obs:memory_order}~(a), $w'$ is before $w$ in memory and there are no other vertices from $B$ between them.
			Each $A_i$ is contiguously stored in memory and disjoint to $B$,
			so it is either whole between $w'$ and $w$ or no part of it is there.
			Let $A'$ be the memory interval bounded by $w'$ and $w$ excluding them.
			If $A'$ is partitioned by $B$, it contains uniquely determined memory-consecutive vertices $u$, $v$ s.t.
			the beginning of $A'$ up to $u$ (inclusive) is to the right of $B$
			and the rest of $A'$, from $v$ on, is to the left of $B$ by (e).
			Any $A_i ⊆ A'$ must contain $u$ and $v$ to be partitioned by $B$.
			If $A'$ is not partitioned by $B$, none of its subintervals is.
	\end{enumerate}
\end{proof}

\image{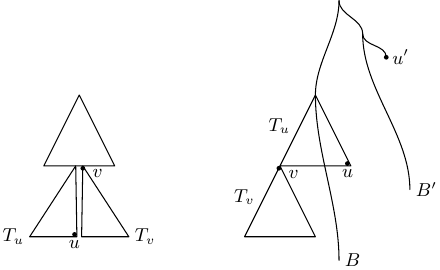}{
	Memory-consecutive vertices $u$, $v$ in tree.
	In the first case (on the left) on~successive branches.
	In the second case (on the right) separated by branch $B$.
	The second case is extended by other branch $B'$ separating new vertex $u'$ from the two.
}

\subsubsection{Subtree in memory}
\label{sect:veb_theory_subtree}

For a vertex $v∈\T$, we define $T(v)∈\D$ as the topmost vertex in $\D$ having $v$ as root.
It is thus the greatest decomposition subtree rooted at $v$.
We also define $h(v)$ as the height of~$T(v)$.
Since the value of $h(v)$ is the same for all vertices on the same level in $\T$,
we also use $h[d] := h(v) = h(T(v))$ for $v$ on level $d$.
We can use $h[]$ as a precomputed array.
The memory layout of a subtree has the properties stated in the following lemma,
see Figure~\ref{img:figure3}:

\begin{lemma}[subtree in memory]\label{lem:subtree_memory}

	Each subtree $T$ is divided into at most $O(\log\log |T|)$ memory intervals,
	each comprising of a range of whole levels of $T$.
	Let $T$ be rooted at depth $d$
	and $\{d=d_0, \dots, d_1-1\},$ $\{d_1, \dots, d_2-1\},$ $\dots,$ $\{d_k, \dots, d_{k+1}-1\}$
	be those ranges.
	Then $d_{i+1} = d_i + h[d_i]$ for $i ∈ \{0, \dots, k\}$.

	The first vertex of the $i$-th memory interval is the leftmost vertex of $T$ on level $d_i$,
	the last vertex of the $i$-th memory interval is the rightmost vertex of $T$ on level $d_{i+1}-1$.
	If $T$ has left sibling,
	the vertex just before the $i$-th memory interval is the rightmost vertex of that sibling on level $d_{i+1}-1$;
	if $T$ has right sibling,
	the vertex just after the $i$-th memory interval is the leftmost vertex of that sibling on level $d_i$.
\end{lemma}
\image{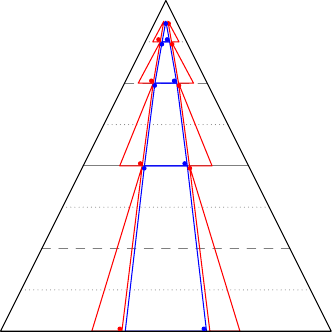}{
	A blue subtree divided into parts which are stored contiguously in different memory intervals.
	The first and last vertex of each such interval is also drawn in blue.
	The~red subtrees are stored just beside the memory intervals of the blue trees
	and have marked their last/first vertices.
}
\begin{proof}
	Let $v$ be the root of $T$ and $T_0 = T(v)$.
	By definition, $T_0$ fills the first memory interval and so $d_1 = d_0 + h(T_0) = d_0 + h[d_0]$.
	The maximality of height of $T_0$ implies that
	either it contains the bottommost level
	or it is at the bottom of the upper subtree of some greater decomposition subtree $R_0$,
	not occupying the whole upper subtree.
	In the latter case, $T$ continues in the bottom part of $R_0$ with the subtrees,
	which are stored contiguously there and occupy the second memory interval of $T$;
	their height is the same as the height of the bottom part of $R_0$,
	which is $h[d_1]$ as there is multiple of them beside each other in memory, so we cannot prolong them.
	Thus $d_2 = d_1 + h[d_1]$.

	Let $T_1$ be the subtree of $T$ containing the first $h[d_0]+h[d_1]$ levels
	and occupying just the first two memory intervals.
	Its height is maximal so it either contains the base level
	or is at the bottom of the upper part of some greater decomposition subtree $R_1$
	and we can proceed with induction.
	Each such step at least doubles the height of $T_i$,
	so $T$ is divided into at most $\cO(\log h[T]) = \cO(\log\log |T|)$ memory intervals.

	By Observation~\ref{obs:memory_order},
	the first vertex in a memory interval has to be the leftmost one on the uppermost contained level
	and the last vertex is the rightmost one on the bottommost level.
	A sibling of $T$ will be divided into memory intervals the same way
	in the same decomposition subtrees $R_i$
	and its memory intervals will be just beside the memory intervals of $T$,
	from which the rest of the second paragraph of the statement follows.
\end{proof}

\subsection{Depth-first search}
\label{sect:veb_dfs}
\subsubsection{Navigating to visit a whole memory interval}
\label{sect:veb_dfs_alg}

Here, we describe
how to navigate during a left-to-right DFS
to visit just vertices of a given memory interval $P$ and their ancestors.

We use a depth-indexed array containing for each tree level the leftmost and rightmost vertex of $P$
and a flag telling whether the currently visited branch is between them.
By Lemma~\ref{lem:interval_in_tree}~(b),
we need to visit just the vertices between the leftmost and rightmost one
and we can easily identify the bounding vertices of all levels by just one scan of the memory interval.

By Lemma~\ref{lem:interval_in_tree}~(c),
the leftmost branch intersecting $P$ contains the leftmost vertex of the bottommost level.
We thus begin our DFS by navigating to the bottommost out of the leftmost vertices,
then we continue to the right until we have visited on each level either both bounding vertices or none of them
and we repeat navigating to another unvisited bottommost bounding vertex if there is any.
We always recognize whether to descend deeper without reading the child vertex
by just comparing its pointer with the bounding vertices (only for equality)
or using the flag saying whether we are between them.
See Algorithm~\ref{alg:RecalculatePointers} for details as part of our pointer recalculation procedure.

\subsubsection{DFS complexity}

Now we bound the I/O complexity of a depth-first search of a contiguous memory interval
as well as DFS of a subtree containing all its leaf descendants,
which is not necessarily contiguous in memory.

We first show that DFS of specific \textit{small decomposition subtrees} occupying at least one block is cache-optimal.
We need only a constant number of blocks to fit in cache for the traversal of those subtrees.
Then we bound the number of visits of most of the~other vertices
by the number of the small subtrees,
so even loading their blocks again for each visit will not worsen the complexity.
Finally, we handle a few remaining vertices on a~constant number of branches
by bounding their visits by search complexity.

Let us define several terms regarding
the state of cache blocks and DFS traversal
at a specific point in time.
In particular, we maintain an \textit{active} branch,
which is moving from left to right according to our DFS traversal.
We also differentiate cache blocks according to whether they intersect active branch
and whether their data were already used. Formally:

\begin{definition}
	The \textit{active branch} of a subtree $T$
		is the leftmost branch so that no already visited vertex is to the right of it.
	A \textit{partially read block}
		is a block which was loaded into memory, but only part of its data were visited.
	Partially read block is \textit{active}
		if it contains any vertex of~the~active branch.
	Partially read block is \textit{inactive} otherwise.
\end{definition}

Now we show the I/O complexity of traversal of the whole small subtrees.

\begin{lemma}[small decomposition subtrees]\label{lem:small_dec_subtrees}
	I/O complexity of DFS of a subtree $T$ in $\D$ of~height in $(⌈\log B⌉, ⌈\log B⌉/ε]$ is $\cO(|T|/B)$.
\end{lemma}
\begin{proof}
	During DFS, we always keep in memory all active blocks;
	by Lemma~\ref{lem:search_complexity}, their number is at most $\cO((\log B/ε) / (ε⋅\log B) + 1) ⊆ \cO(1/ε^2)$.
	When active branch is being changed, we will not reread blocks that are already active;
	thus we need to read a block again only after it was inactive.

	If a partially read block becomes inactive,
	it is partitioned by an active branch into two components
	and, by Lemma~\ref{lem:interval_in_tree}~(f), this may happen in at most one way.
	Each block is thus read at~most twice.

	The tree $T$ is contiguously stored in memory
	and uses at least one whole block due to arity of~inner vertices being at least two.
	In total we need to read at most $2(|T|/B + 2) ∈ O(|T|/B)$ blocks.
\end{proof}

We also need to handle small subtrees which are stored in our memory interval only partially.
Either the memory interval is small enough that it is fully contained in the subtree,
or we have at most two such subtrees at the borders of the larger memory interval.

\begin{lemma}[parts of small decomposition subtrees]\label{lem:parts_of_small_dec_subtrees}
	Let $T$ be a subtree in $\D$ of height at~most $⌈\log B⌉/ε$, let $P$ be its memory-contiguous part,
	then DFS of $P$ has I/O complexity $\cO(|P|/B + (\log_B |T|)/ε + 1)$ starting the DFS in the root of $T$.
\end{lemma}
\begin{proof}
	We mark visited vertices by one of the three colors: blue, green and red. (see Figure~\ref{img:figure4})
	The vertices of $P$ are all blue;
	by Lemma~\ref{lem:interval_in_tree}~(f),
	they may be divided into at most two distant branch-separated components.
	The non-blue ancestors of blue vertices in each component are green
	except for that on the leftmost and rightmost branch, which are red.
	So we have at most four red subbranches.

	As in the proof of the previous lemma, we keep all active blocks in memory and load them at~most twice.
	The blue part is still contiguous part of memory,
	but this time it may be smaller than one block;
	visiting the blue vertices thus has I/O-complexity $\cO(|P|/B + 1)$.

	The green parts are not necessarily contiguous in memory,
	but after having loaded all active blocks,
	we can just continue DFS as in the previous lemma.
	So we first load $\cO((\log_B |T|)/ε + 1)$ blocks on the left red branch as one search in tree,
	then continue visiting the green vertices,
	and finally pay another $\cO((\log_B |T|)/ε + 1)$ for the right red branch,
	which may contain some partially used blocks.
	The number of fully green blocks is lower than the~number of blue blocks,
	because green vertices have arity at least 2 in DFS;
	so we can pay their loading by the blue part.

	Now only partially green blocks without blue and red vertices remain.
	These blocks need to be partitioned by one of the four red branches
	having one component green and the other unvisited.
	By Lemma~\ref{lem:interval_in_tree}~(g),
	we have at most one such block for each red vertex,
	and moreover it is in memory between the vertex and its parent.
	Therefore we count only the~divided blocks corresponding
	to the red vertices not sharing their blocks with their parents.
	This leads to $\cO((\log_B |T|)/ε + 1)$ blocks.
\end{proof}

\image{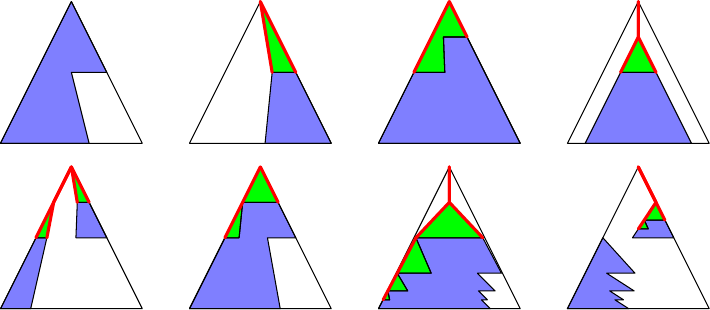}{
	Different cases of memory interval in tree.
	The memory interval is blue,
	ancestors of~blue vertices are either green, or red if on boundary.
	See the proof of Lemma~\ref{lem:parts_of_small_dec_subtrees} for more details.
}

We now compute the I/O complexity of the whole DFS and then of PMA updates.

\begin{theorem}[DFS complexity]\label{thm:dfs}
	Visiting all vertices in the contiguous part $P$ of memory using our DFS algorithm
	has I/O complexity $\cO(|P|/B + (\log_B N)/ε + 1)$;
	the time complexity is $\cO(|P| + \log N)$.
	The I/O and time complexity of a DFS of a subtree $T$ containing all descendants of its root is
	$\cO(|T|/B + ε^{-2})$ and $\cO(|T|)$, resp.,
	if we begin in the subtree root.
	Using depth-indexed arrays accessed by the current depth during the DFS or sequential scan of $P$
	will not worsen the complexities.
\end{theorem}
\begin{proof}
		We begin with the first part of the theorem and
		mark all visited vertices as blue, green and red
		the same way as in the previous lemma
		except for a branch leading to the rightmost vertex of the bottommost level,
		which will be red (see Figure~\ref{img:figure5}).
		We note that each vertex of arity one in DFS is now red.

		If the tree is too low, we can use directly Lemma~\ref{lem:small_dec_subtrees} or~\ref{lem:parts_of_small_dec_subtrees}.
		Let us assume $h(\T) > ⌈\log B⌉/ε$.

		We take all the bottommost visited vertices,
		i.e. visited vertices not having their children visited,
		and find lowest decomposition subtrees of height at least $⌈\log B⌉$ containing them.
		If any such subtree is part of another one, we use only the higher.
		We call them \textit{small subtrees}.
		All but two of these subtrees are fully blue (dark blue in Figure~\ref{img:figure5});
		the at most two other are on~the~boundaries of $P$.
		We further divide small subtrees into \textit{leaf} subtrees and \textit{non-leaf} subtrees,
		where leaf ones have no visited descendants.

		Each small subtree has height at most $⌈\log B⌉/ε$.
		We use Lemmas~\ref{lem:small_dec_subtrees} and~\ref{lem:parts_of_small_dec_subtrees} to determine the~I/O complexity of all leaf small subtrees;
		the sum is certainly $\cO(|P|/B + (\log_B N)/ε + 1)$, where the $\cO((\log_B N)/ε + 1)$ part is for the at most two partially blue small subtrees.
		As~the~scan of~leaf subtrees is not interrupted by descending below them, the lemmas can be used.

		The number of ancestors of leaf small subtrees of degree at least two in DFS is lower than the number of these subtrees
		and the sum of their degrees is lower than twice the number of~the~subtrees.
		Even if we would need to load whole cache block for each visit of each of these ancestors,
		the number of loads will be asymptotically bounded by the number of leaf small subtrees.
		As each fully blue small subtree occupies at least one block,
		we can bound this part by $\cO(|P|/B + (\log_B N)/ε + 1)$.

		We now bound visits of the ancestors of degree 1, all of which are on a constant number of red branches, by $\cO((\log_B N)/ε + 1)$.
		We divide the branches into segments occupying individual cache blocks.
		If a segment contains only degree-1 vertices, we need to load it at~most twice~-- when going down and back up.
		If it contains vertices of higher degrees, further reloadings of~the~block are payed by the reasoning in~the~previous paragraph.

		Finally, we consider the non-leaf subtrees, which is the last part of the tree to be paid for (see Figure~\ref{img:figure6}).
		By Lemma~\ref{lem:interval_in_tree}~(c),
		the leaves of such a subtree with visited descendants in~$\T$ are all to~the~left of all leaves without such descendants.
		We note that the rightmost branch of~the~left part is fully red, as it leads to the rightmost vertex on the bottommost level of~$P$.
		We have already paid for loading vertices in the left part as described in~the~previous paragraphs.
		After visiting them and their descendants,
		we load all the vertices on the red branch of this small subtree and continue traversing it according to Lemma~\ref{lem:small_dec_subtrees}.
		Traversing the~right part of~the~subtree is thus bounded by traversing the whole small subtree in case it had no descendants.
		The~whole I/O complexity of this part is then also bounded by $\cO(|P|/B + (\log_B N)/ε + 1)$.

		The depth-indexed arrays accessed by the current depth
		such as the stack of ancestors or arrays for navigation
		will not worsen the complexity,
		because loading their blocks into memory can be bounded
		by loading the corresponding part of the active branch.
		Accessing those arrays during sequential scan of $P$ hides in the complexity of the scan:
		We focus on the lowest level of decomposition where all trees are larger then $B$;
		during the scan, we read sequence of those subtrees and each requires $O(1)$ blocks of depth-indexed arrays.

		The total I/O-complexity sums up to $\cO(|P|/B + (\log_B N)/ε + 1)$.
		Similarly we get the~stated time complexity.

		To prove the second part of the theorem, we proceed similarly.
		This time, $T$ is not neccessarily contiguous in memory,
		so we mark only the bottom contiguous part as blue
		and the remaining ancestors as green, which will be nevertheless also visited.
		We note that $(\log_B|T|)/ε$ in
		Lemma~\ref{lem:parts_of_small_dec_subtrees} is $\cO(ε^{-2})$,
		so we get $\cO(|T|/B + ε^{-2})$ for the small subtrees.
		To get rid of the $\cO((\log_B |T|)/ε)$ arising from the red branches,
		we observe that all ancestors of the blue vertices that are part of $T$
		are now of degree at least 2 and we have no non-leaf small subtrees,
		so the red branches are not needed at all.
\end{proof}

\image{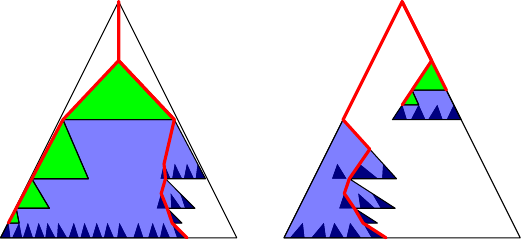}{
	A tree containing blue memory interval with dark blue small subtrees.
	The ancestors of~blue vertices are green.
	Several special branches are red.
}

\image{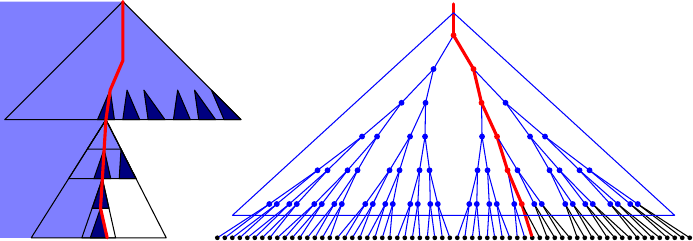}{
	On the left, non-leaf small subtrees crossed by the red branch.
	On the right, detail of~one non-leaf small subtree crossed by the red branch.
}
\subsection{Building the (sub)tree using DFS}
\label{sect:veb_build}

Assuming we have enough space reserved for the vertices of a new subtree
in the right memory intervals,
we show how to build there the subtree, i.e. initialize vertices with pointers to their children.
If we are building the whole tree, there is just one such memory interval;
if we are building its subtree, there may be multiple of them as shown by Lemma~\ref{lem:subtree_memory}.
The lemma also shows that we can find where to reserve the space for the subtree by traversing
a single branch of its existing sibling subtree, which we use later in Section~\ref{sect:corobts};
here, we do not count the required memory movements into complexity.

We first present a general approach how to navigate in a non-existent tree $T$ using depth-first search,
which can be easily used to build a fixed-arity tree.
Then, we generalize it in Section~\ref{sect:veb_general_build}
to build an arbitrary subtree provided that we can simulate its DFS.

Let $h[]$ be the precomputed array as defined in Section~\ref{sect:veb_theory_subtree}
and $M[i][j]$ be (pointer to) the $j$-th vertex of the $i$-th memory interval.
We use arrays $I$, $J$, $K$ indexed by depth for navigation.
For a~depth $d$, $M[I[d]][J[d]]$ is either the last visited vertex on level $d$ or the next vertex to~visit;
it~will be the latter each time we are descending to that level,
so we can safely initialize pointer from its parent.
In other words, array $I$ determines the correspondence between levels in $T$ and memory intervals containing them,
and so it stays constant throughout the~whole process;
array $J$ points to the current positions in the respective memory intervals.
The~value $K[d]$ is the number of times vertex $M[I[d]][J[d]]$ was visited
and it is reinitialized to zero each time $J[d]$ is changed.
We use it to determine whether to descend to next child or go up.

The topmost vertex of each memory interval on the leftmost branch is the first vertex in that interval,
by Lemma~\ref{lem:subtree_memory}.
We calculate indices of vertices in $M$ which are in the same memory interval as their parents and are their leftmost children
from indices of their ancestors in the interval (see green arrows on Figure~\ref{img:figure7}).
The indices of other vertices are calculated from the vertices just before them in memory,
which were already visited (see red arrows on Figure~\ref{img:figure7}).

\image{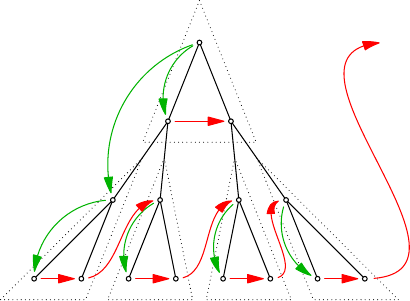}{
	Building subtree.
	Arrows show how indices of new vertices are computed.
}

In a fixed-arity case,
we can calculate the size of each decomposition subtree just from its height.
Let $R$ be a decomposition subtree rooted at $v$ and $R_0$ its leftmost child in $\D$ (upper part of~$R$).
We can calculate size of $R_0$ and so the index of the root $w$ of the second leftmost child~$R_1$ of~$R$ in~$\D$.
The vertex $w$ is on the leftmost branch of $R$, just under $R_0$.
After entering a~vertex~$v$, we always calculate positions of such $w$ for each decomposition subtree $R$ rooted at $v$.
This ensures that the position of the leftmost child $w'$ of any vertex is always initialized before descending to it:
either the parent of $T(w')$ in $\D$ is whole in one of~the~memory intervals and its root was already visited,
or $w'$ is the first vertex in its memory interval.

After returning to a vertex $v$ from its child $w$ willing to descent to another child $w'$,
the~tree~$T(w)$ is already built and the vertex $w'$ is in memory just after the rightmost leaf of~$T(w)$,
which is the last visited vertex on level $d(v) + h[d(v) + 1]$.
Indices of this vertex are still stored in $I$,~$J$,
so we can compute the index of $w'$ and descend there.

See Algorithm~\ref{alg:BuildSubtree} for details.

\bigskip
The I/O and time complexity is bounded by the DFS complexity during pointer recalculation:

\begin{lemma}[complexity of subtree building]\label{lem:subtree_build}
	The I/O complexity of building a given-arity subtree $T$ is
	$\cO(|T|/B + \min((\log_B N)/ε + 1,\, ε^{-2}))$.
	The time complexity is~$\cO(|T|)$.
\end{lemma}
\begin{proof}
By~Theorem~\ref{thm:dfs}, the I/O complexity of a DFS of the tree $T$
can be expressed as $\cO(|T|/B + ε^{-2})$ or $\cO(|T|/B + (\log_B N)/ε + 1)$
if considered as the DFS of the last memory interval of $T$,
which visits the same area.
We only need to show that accessing the arrays $I$, $J$, $K$ do not worsen the I/O complexity.

We access the arrays $I$, $J$, $K$ in two cases.
The first one is when we access the currently visited vertex,
which is taken into account in~Theorem~\ref{thm:dfs}.
The~second case is accessing the more distant vertex of the green and red arrows in Figure~\ref{img:figure7};
we~further analyze this case.

We partition $T$ into the smallest possible decomposition subtrees of heights at least $⌈\log B⌉$.
The~topmost one may contain even vertices outside of $T$;
the others are partitioning of~the~rest.
Each of those subtrees occupies at least one cache block
and parts of $I$, $J$, $K$ corresponding to~the~branch of a specific subtree still fit in cache at the same time.
We distinguish short (red/green) arrows inside individual subtrees and long arrows between different subtrees.

When calculating indices on the leftmost branch inside a subtree (short green arrows),
we can keep them in memory to the time we visit them.
Similarly, when calculating indices from already visited vertices inside a subtree (short red arrows),
they might have been kept in memory from the time they were visited.
No extra blocks are then loaded for these cases.

When calculating indices of other vertices (long arrows) they are always roots of~the~subtrees.
We pay for loading the more distant vertex by the subtree.
\end{proof}

\subsubsection{Building an arbitrary tree}
\label{sect:veb_general_build}

A general procedure for subtree build is as follows.
First, we simulate DFS of the resulting subtree $T$
and store all its vertices contiguously in a temporal array $L$ in left-to-right preorder;
we require this simulated walk to be possible.
Each vertex contains its data, depth and a constant-sized space to be used later,
but no pointers.
Then, we scan $L$ backwards, which corresponds to right-to-left postorder,
count sizes of the decomposition subtrees and write them to the relevant vertices in $L$.

We use arrays $I$, $J$, $K$, $M$ the same way as in the fixed-arity tree build.
The difference is that the arity of individual vertices vary
and so we cannot easily compute the sizes of decomposition subtrees,
which are then used to obtain memory positions of vertices pointed
by green arrows in Figure~\ref{img:figure7}.
The precise goal of our backward scan of $L$ is
that each vertex pointed by the green arrow contains
the size of the largest decomposition subtree rooted in the arrow's origin
not containing the arrow's destination.
This is exactly the memory offset needed during the subtree build.
To build the subtree, we scan $L$ in forward direction,
infer from depths of following vertices whether the current vertex have other children
and read the desired sizes of subtrees when needed.

It remains to show how the backward scan of $L$ works.
Let $D$ be another depth-indexed array.
It aggregates the number of vertices in the largest decomposition subtrees having their bottommost level at the given depths.
The height of such subtrees for each depth is given by the decomposition.
Each time we reach the root of such a subtree,
we copy the value to a vertex in $L$,
possibly add it to a deeper counter
and zero it.
See Figure~\ref{img:figure8}.

\image{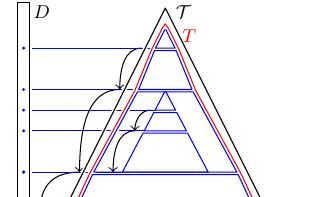}{
	General subtree build: schema of summing counters in the array $D$ during backward scan of $L$.
	The subtree $T$ stored in $L$ is drawn red as part of the whole tree $\T$, where it should be build.
	The array $D$ is drawn to the left with blue dots representing some of the counters.
	The counters are connected to the blue areas in the tree $T$
	whose vertices should be counted there after they are whole visited.
	Other areas are added earlier or later as illustrated by the black arrows.
}

The exact procedure of visiting a vertex $v$ in the backward scan of $L$ is as follows.
Let $v$ be the root of decomposition subtrees of heights $h_0 = 1, h_1, \dots, h_k$
and let $v_i$ for $i ∈ \{0, \dots, k-1\}$ be the leftmost vertex just under the decomposition subtree of height $h_i$.
We note that since $L$ contains vertices in preorder,
the whole leftmost branch of the subtree rooted at $v$ is stored consecutively
and the vertex $v_i$ is $h_i$ steps after $v$.
First, we increment the counter $D[d(v)]$ on the vertex level.
Then for each $i ∈ \{0, \dots, k-1\}$,
we copy the size of $i$-th decomposition subtree $D[d(v) + h_i - 1]$ to vertex $v_i$,
add that value to the counter of the next decomposition subtree $D[d(v) + h_{i+1} - 1]$ and zero it.
See Figure~\ref{img:figure9} for the illustration of the algorithm.

\image{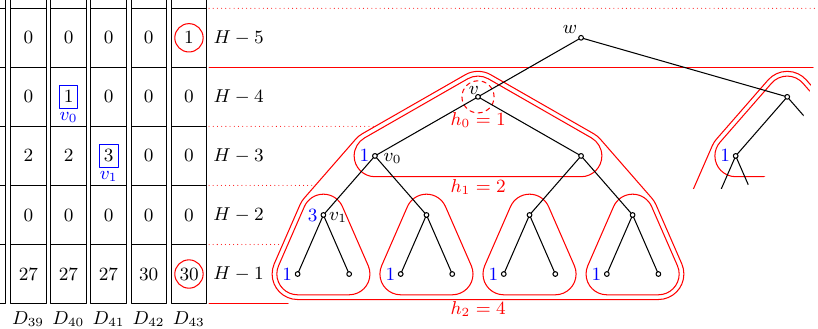}{
	General subtree build: backward scan of $L$ in detail.
	The subtree $T$ whose preorder is written in $L$ is rooted at $w$.
	The red tree boundaries represent the vEB decomposition, blue numbers are the sizes of subtrees,
	which should be written in the vertices at the end.
	Each column of the table to the left represents the state of the array $D$
	at a specific point in time.
	The column $D_{39}$ is the state when the right-to-left postorder scan enters $v$:
	At the leaf level $D_{39}[H-1]$ contains number (27) of vertices in the right child subtree of $w$ (15)
	and the four bottom 3-vertex subtrees in the left subtree (12);
	$D_{39}[H-3]$ counts the children of $v$.
	After entering $v$, its counter $D[H-4]$ is incremented ($D_{40}$).
	Then it is copied to $v_0$,
	added to counter $D[H-3]$ and zeroed ($D_{41}$).
	The counter $D[H-3]$ is copied to $v_1$, added to counter $D[H-1]$ and zeroed ($D_{42}$).
	After entering $w$, its counter $D[H-5]$ is incremented and the scan finishes ($D_{43}$).
	The two values in red circles (1, 30) are the sizes of the resulting memory intervals.
}

We never zero the counter of the largest decomposition subtree,
because it may be shared with some greater subtree; the value is not needed during build.
After the backward scan finishes, some non-zero values are left in $L$,
each corresponds of the desired size of the memory interval containing that level as the bottommost.
The first non-zero value contains the size of the largest decomposition subtree $T_0$ rooted at the root of the subtree $T$ to be built;
the tree $T_0$ will be stored in the first memory interval.
The second non-zero value counts vertices of the decomposition subtrees under $T_0$
as the size of $T_0$ was not added yet to the value;
it thus corresponds to the second memory interval, etc.

\begin{lemma}[complexity of general subtree building]\label{lem:general_subtree_build}
	The I/O complexity of building a general subtree $T$ is
	$\cO(|T|/B + \min((\log_B N)/ε + 1,\, ε^{-2}))$,
	provided that we can simulate its DFS.
	The time complexity is~$\cO(|T|)$.
\end{lemma}
\begin{proof}
	The array $L$ is accessed sequentially except for the destinations of green arrows
	in Figure~\ref{img:figure7}, which is bounded by accessing the same vertices in the tree.
	During the backward scan the same vertices are accessed out of order and the same argument applies;
	the array $D$ is used similarly to $I$, $J$, $K$.
	Lemma~\ref{lem:subtree_build} applies to the rest.
\end{proof}

\section{Packed-Memory Array}
\label{sect:pma}
\subsection{Original Construction}
\label{sect:original_pma}

Each cell of PMA $\P$ either contains an~element or is \textit{blank} and elements are stored in $\P$ in the same order as in the original sequence.
We conceptually split $\P$ into \textit{segments} of size $\Theta(\log |\P|)$ so that their number is a power of two.
We imagine a~perfectly balanced binary tree with leaves corresponding to the segments.
We call its vertices \textit{nodes} not to confuse them with vertices stored inside PMA.
Each inner node then corresponds to a subarray of $\P$.
We define the \textit{capacity} of a node as the~size of its subarray
and the \textit{density} as the ratio between the number of stored elements in the subarray and the capacity.
Now we impose a density constraint on the nodes based on their depth;
both lower and upper bounds are getting more strict when going upwards.
In~particular, we arbitrarily define constant leaf-level bounds and more strict root-level bounds
and then linearly interpolate the intermediate levels according to their number.

When inserting or removing an element,
the simplest case is if the resulting density of the corresponding segment still meets its constraint.
We then \textit{rebalance} only that segment by~evenly distributing all the elements within it.
In the other case, we find the~bottommost node satisfying its constraint and rebalance that node,
which fixes all the~constraints in~its~subtree.
Since the tree is only virtual and nowhere stored,
we calculate the densities of~nodes by incrementally scanning the array, in both directions.
The final uniform redistribution of the~elements also requires a constant number of sequential scans.
If the density in~the~root is outside of its interval, we rebuild the whole PMA.
The worst-case I/O complexity of~the~insert/delete operation is $\cO(K/B + 1)$,
where $K$ is the size of the subarray to be rebalanced; the~time complexity is $\cO(K)$.
The aforementioned amortized I/O complexity $\cO((\log^2 N)/B + 1)$ and time complexity $\cO(\log^2 N)$ follow from Lemma~\ref{lem:pma_amort}.
The proof is a~reformulated version of~the~analysis of Bender et al.~\cite{BDFC05}.

\begin{lemma}[PMA amortization]\label{lem:pma_amort}
	Let $K$ be the capacity of the PMA node to be rebalanced
	after violating the density constraint of its child during an insert/delete operation,
	or the~capacity of the root in the case of total rebuild.
	Let the worst-case complexity of the operation be $\cO(CK)$ for an arbitrary $C>0$ fixed throughout the existence of the PMA.
	Then the~amortized complexity is $\cO(C⋅\log^2 N)$.

	The amortized complexity of initialization with an arbitrary number of initial elements (or total rebuild)
	can be bounded by the worst-case complexity as it zeroes the potential.
\end{lemma}
\begin{proof}
	Let $h$ be the height of the PMA conceptual tree,
	the density of nodes on level $k$ is within the interval $[\rho_k,\tau_k]$,
	and the density of the root after total rebuild is within the interval $[\rho_{-1}, \tau_{-1}]$.
	The following is satisfied:
	$0<\rho_h<\rho_0<\rho_{-1}≤\tau_{-1}<\tau_0<\tau_h≤1$,\penalty -10000
	$\rho_k = \rho_0 - (\rho_0 - \rho_h)⋅k/h$ and
	$\tau_k = \tau_0 + (\tau_h - \tau_0)⋅k/h$ for $0<k<h$.

	We define a potential
	\[
		\Phi=cC⋅\log N ⋅\sum_{\textrm{node $n$}} \f{cap}(n) ⋅ \left\{
		\begin{array}{ll}
			\rho_{d(n)-1} - \rho(n)
				& \f{for } \rho(n) < \rho_{d(n)-1}, \\
			0
				& \f{for } \rho_{d(n)-1} ≤ \rho(n) ≤ \tau_{d(n)-1}, \\
			\rho(n) - \tau_{d(n)-1}
				& \f{for } \tau_{d(n)-1} < \rho(n), \\
		\end{array}
		\right.
	\]
	where the sum is over all PMA nodes,
	$\rho(n)$ is the density of the node $n$, $\f{cap}(n)$ its capacity, $d(n)$ its depth,
	and $c$ sufficiently large constant.

	Rebalancing can only decrease the contribution of each node in $\Phi$.
	The topmost node~$n$ originally violating its density constraint contributed before rebalancing at~least
	\[
		cC⋅\log N ⋅ \f{cap}(n) ⋅ \min(\rho_{d(n)-1} - \rho_{d(n)}, \tau_{d(n)} - \tau_{d(n)-1}) ∈ Ω(C⋅\f{cap}(n)) = Ω(CK);
	\]
	after rebalancing it decreases to zero,
	which pays for the operation.

	On the other hand, insertion or removal of one element increases the potential by at~most $cC⋅\log N$
	per each of the $\cO(\log N)$ PMA nodes containing the element.
	It results in~the~potential increase of $\cO(C⋅\log^2 N)$,
	which is the amortized complexity of the operation.
\end{proof}

Leo and Boncz briefly described how to modify PMA to support batch updates~\cite{LB19},
so that we can give it a sequence of insert/remove operations to be performed during one pass.%
\footnote{We will, actually, not mix inserts with removals in a batch, but there is no extra cost for allowing it.}
We use a similar approach extended by updating the pointers to items which are moved.

\subsection{Our Version}
\label{sect:our_pma}

We now introduce the changes to the PMA mentioned in Section~\ref{sect:overview},
so that it allows batch updates and calls our pointer recalculation procedure on vertices which will be moved.
The I/O and time complexity of PMA itself excluding the pointer recalculation remain asymptotically the same
since the amortized analysis from~Lemma~\ref{lem:pma_amort} still applies and each individual insert/remove pays still the same amount into the potential.

\bigskip
The input of our batch update algorithm consists of several insert/remove operations sorted according to their positions in memory.
An insert operation is a pair of a position in PMA and a number of items to be inserted there --
for each such insert operation the PMA update returns an array of pointers to the new items to be filled with the actual data.
A remove operation is defined by a range of elements to be removed.
All the lists we use in the following text are in fact arrays or extendable vectors if appending is needed.
We divide PMA updates into four phases.

In the first phase,
we identify all the consecutive parts of memory, \textit{memory intervals},
corresponding to the conceptual tree nodes
in~which items will be uniformly redistributed.
We start calculating the densities of nodes from the first place to be updated as usual.
During forward scans we take into account other updates occurring in that part of memory
and during backward scans we merge our memory interval with the previous one if we reach their common parent node.
After obtaining a memory interval we continue with the next unvisited update operation.
The output of this phase is the list of memory intervals for update
along with the final number of items in each of them.

In the next phase, we scan all the intervals again and compute new positions of the contained items including the newly inserted ones.
For the old items, we store the pointers to their new positions as part of them --
we keep a reserved space inside the items for this purpose.
For the new items, we create the arrays containing pointers to their final positions.
We note that none of the pointers are valid at this point since no movements have occurred yet.

Next, we call pointer recalculation for each interval.
It returns a list of changes in~pointers to~children inside the tree $\T$
to make the pointers valid after moving the vertices in~the~interval.
We will expand on pointer recalculation in~Section~\ref{sect:corobts_recalc}
and we will also analyze its impact on complexity there.
Briefly speaking,
amortization applies to the operations with the same I/O complexity as sequential scans;
other operations will be considered separately.

In the last phase, we move all the items to their final positions and
perform removals by several sequential scans.
Then, we update the pointers between the items according to the list from the previous paragraph;
we process the list in the same order it was generated,
so the complexity is upper-bounded by the complexity of the pointer recalculation phase.
Finally, we initialize newly created items.

As a special case, the need for extending or shortening the memory may arise
during the initial intervals identification
to keep the size of $\P$ within a constant factors of the number of items.
In that case, we just distribute the elements uniformly in the newly allocated memory
and recalculate pointers to~all of~them.
This will not affect the complexity when choosing the density bounds correctly.
For more details see Algorithms~\ref{alg:PMABatchUpdate}, \ref{alg:PMAGetIntervals} and~\ref{alg:PMACalcNewPositions}.

We get the following complexity of PMA updates excluding pointer recalculation;
we compute the total complexity of PMA in Theorem~\ref{thm:pma+dfs}.

\begin{lemma}[PMA complexity]\label{lem:pma}
	Consider a batch PMA update consisting of $L$ operations
	where each is either an insert of multiple elements at one place
	or a removal of multiple consecutively stored elements;
	let $S$ be the total number of items to be inserted or removed.

	Then the updated memory is divided into at most $L$ memory intervals
	and for each of them pointer recalculation is called.
	The total I/O complexity excluding the pointer recalculation phase is
	$\cO(S⋅(\log^2 N)/B + L)$ amortized,
	which follows from the worst-case complexity
	$\cO(K/B + L)$,
	where $K$ is the sum of sizes of the affected memory intervals.
	The total time complexity excluding the pointer recalculation phase is
	$\cO(S⋅\log^2 N)$ amortized or $\cO(K)$ worst-case.

	The I/O complexity of the initialization with $N$ initial elements is $\cO(N/B + 1)$,
	the time complexity is $\cO(N)$ and no pointer recalculation is needed.
\end{lemma}

We note that the novelty of the statement is in the bounded number of memory intervals
for which pointer recalculation is called.
In the context of a persistent array, we use a little different form:

\begin{lemma}[PMA complexity 2]\label{lem:pma2}
	Consider a batch PMA update consisting of $L$ operations
	where each is either an insert of multiple elements at one place
	or a removal of multiple consecutively stored elements;
	let $S$ be the total number of items to be inserted or removed.

	Let us choose $L$ items of PMA containing all places for insertions and one element from each interval for removal.
	Let all these items be contained within $C$ cache blocks.
	Then we can express the~amortized I/O complexity as $\cO(S⋅(log^2 N)/B + C)$.
	If all these items along with $Θ(B)$ space before and after each of them are preloaded in the memory,
	then the amortized complexity is only $\cO(S⋅(\log^2 N)/B)$.
	Both are considered without the pointer recalculations.%
\end{lemma}

\section{The \corobts{} Data Structure}
\label{sect:corobts}
In this section,
we combine the van Emde Boas memory layout of Section~\ref{sect:veb}
with the packed memory array of Section~\ref{sect:pma}
to obtain a dynamic data structure for storing $(a,b)$-ary trees.

We store the tree in the vEB layout in PMA.
Each vertex $v$ stores its depth $d(v)$,
pointers to children, reserved space for the pointer to its new position used during PMA updates,
and other constant-sized data such as keys for searching.
We can observe that the root of $\T$ as an entry point is the~first vertex in $\P$.

A general procedure for an update is as follows:
We let PMA reserve enough space for new vertices at right places in memory if needed.
Then we perform tree changes such as (re)building or removing a subtree.
Finally, we let PMA remove memory intervals containing removed vertices if there were any.
The PMA update itself first plan the necessary memory movements
and tells us which vertices should be moved and where,
we then update the pointers leading there
and let PMA perfom the movements.
Except for subtree rebuild, PMA is updated just once per tree update.

In Section~\ref{sect:corobts_recalc},
we first show how the pointer recalculation as part of the PMA updates works
and prove the complexity of PMA updates including pointer updates.
Then we describe and analyze initialization, subtree insertion and removal and subtree rebuild
in Sections~\ref{sect:corobts_init}, \ref{sect:corobts_insert_remove} and~\ref{sect:corobts_rebuild},
respectively.

\subsection{Pointer Recalculation within PMA updates}
\label{sect:corobts_recalc}

When the PMA needs to move the vertices within a single memory interval,
which may occur multiple times in sequence during the PMA update,
it tells us the new positions and we need to update the pointers leading there.
Namely, we get a memory interval $P⊆\P$ whose vertices will be moved and
pointers to their new positions stored inside the vertices.
We also get a preallocated array, in which we will append a list of changes in pointers to children to be performed after all the movements.
Recall that we have no pointers to parents,
so we have to traverse the tree from the root to be able to modify pointers leading to children in that contiguous memory.

We use depth-first search in left-to-right tree order
visiting just the vertices in $P$ and their ancestors;
we described the navigation in the tree in Section~\ref{sect:veb_dfs_alg}.
During DFS, we maintain a stack $S$ (stored as array) of all ancestors of current vertex,
so~we can easily return there without parent pointers.
For each visited vertex in $P$, we update the~pointer leading there from its parent,
or more precisely we just schedule this update by appending it to the preallocated array, not~to~break the current tree structure.
The actual updates will be performed later in the same order we process vertices during DFS.
See Algorithm~\ref{alg:RecalculatePointers} for details.

\begin{theorem}[PMA+DFS complexity]\label{thm:pma+dfs}
	Consider a batch PMA update consisting of $L$ operations
	where each is either an insert of multiple elements at one place
	or a removal of~multiple consecutively stored elements;
	let $S$ be the total number of items to be inserted or removed.
	Then the total amortized I/O complexity is $\cO(S⋅(\log^2 N)/B + L⋅(\log_B N)/ε + 1)$
	and the total amortized time complexity is $\cO(S⋅\log^2 N)$.
\end{theorem}
\begin{proof}
	By Lemma~\ref{lem:pma},
	the pointer recalculation is called on at most $L$ memory intervals whose sum of~sizes is $K$,
	as defined in the lemma.
	The worst-case I/O complexity of all the pointer recalculations is then $\cO(K/B + L⋅(\log_B N)/ε + L)$, by Theorem~\ref{thm:dfs}.

	We can use the amortization argument from Lemma~\ref{lem:pma} on $\cO(K/B + L)$ part of the~previous complexity,
	which yields the amortized I/O complexity of all pointer recalculations $\cO(S⋅(\log^2 N)/B + L⋅(\log_B N)/ε + L)$.
	It also asymptotically bounds the rest of PMA updates from the same lemma and so represents the total I/O complexity.
	We can replace the~last term, $L$, by 1,
	because it would be greater than $L⋅\log_B N$ only if $N<B$ and so the whole data structure fits into constant number of blocks.
	Similarly, we get the time complexity $\cO(S⋅\log^2 N + L⋅\log N) ⊆ \cO(S⋅\log^2 N)$.
\end{proof}

\subsection{Initialization}
\label{sect:corobts_init}

\textit{Initialization} of a new $a$-ary tree requires just its height as an argument;
to initialize a general tree the simulation of its DFS has to be provided according to the Section~\ref{sect:veb_general_build}.
We just initialize PMA with enough space and build there the tree as described in Section~\ref{sect:veb_build}.
See Algorithm~\ref{alg:Init} for details.

\begin{theorem}[initialization complexity]\label{thm:init}
	The I/O complexity of data structure initialization is $\cO(N/B + 1)$,
	where $N$ is the size of the whole tree.
	The time complexity is linear.
\end{theorem}
\begin{proof}
	The I/O complexity of PMA initialization is $\cO(N/B + 1)$ and the time complexity is linear by
	Lemma~\ref{lem:pma}.
	The I/O complexity of subtree build is
	$\cO(N/B + (\log_B N)/ε + 1) ⊆$ $\cO(N/B + 1)$
	and time complexity is linear
	by Lemma~\ref{lem:subtree_build} for fixed-arity tree
	or Lemma~\ref{lem:general_subtree_build} for general tree.
\end{proof}

\subsection{Subtree insertion and removal}
\label{sect:corobts_insert_remove}

\textit{Inserting} an $a$-ary subtree $T$ requires its parent vertex and the index between its siblings,
where $T$ has to be built, it can be filled by other data afterwards;
inserting a general subtree according to its simulated DFS is also possible.
We first traverse the rightmost branch of the left sibling of $T$
or the leftmost branch of the right sibling of $T$
to find where in memory the space for the memory intervals of $T$
has to be reserved according to Lemma~\ref{lem:subtree_memory}.
The lemma also tells us how many levels each memory interval contains,
so we can compute their sizes in case of an $a$-ary tree;
for the general case the computation is described in Section~\ref{sect:veb_general_build}.
Then, we call PMA update to reserve the needed space.
Finally, we build there the new subtree as described in Section~\ref{sect:veb_build}.
See Algorithm~\ref{alg:InsertSubtree} for the $a$-ary tree insert.

\textit{Removing} a subtree $T$ requires just its root vertex.
We first traverse the leftmost and rightmost branches of $T$
to identify its memory intervals according to Lemma~\ref{lem:subtree_memory}.
Then, we remove the pointer to $T$ from its parent
and call PMA update to remove the memory intervals.
See Algorithm~\ref{alg:RemoveSubtree} for details.

\begin{theorem}[subtree insert/remove complexity]\label{thm:io}
The amortized I/O complexity of subtree insertion or subtree removal
is $\cO(S⋅(\log^2 N)/B + (\log_B N ⋅ \log\log S)/ε + 1)$,
where $S$ is the size of the subtree to be created or removed and $N$ the size of the whole tree $\T$ after insertion or before removal.
The amortized time complexity is~$\cO(S⋅\log^2 N)$.
\end{theorem}
\begin{proof}
	Traversing a branch has the I/O complexity $\cO((\log_B N)/ε + 1)$ by Lemma~\ref{lem:search_complexity}.

	The subtree may be divided into $\cO(\log\log S)$ memory intervals by Lemma~\ref{lem:subtree_memory}
	and so the I/O complexity of PMA update is
	$\cO(S⋅(\log^2 N)/B + (\log_B N ⋅ \log\log S)/ε + 1)$
	and time complexity $\cO(S⋅\log^2 N)$
	by Theorem~\ref{thm:pma+dfs}.
	In case of removal, we consider the complexity of the DFS in the pointer recalculation
	as if the subtree for removal is still there to avoid analysis of the holes in the memory.

	The subtree build has the I/O complexity
	$\cO(S/B + (\log_B N)/ε + 1)$ and linear time complexity
	by Lemma~\ref{lem:subtree_build} for fixed-arity subtree
	or Lemma~\ref{lem:general_subtree_build} for general subtree.
\end{proof}

\subsection{Subtree rebuild}
\label{sect:corobts_rebuild}

\textit{Rebuilding} a subtree uses our general build procedure described in Section~\ref{sect:veb_general_build}.
We first simulate DFS of the resulting subtree, which we require to be possible,
and create the array $L$ containing the vertices in preorder.
Then we scan all the memory intervals,
insert new vertices to intervals where they are missing,
and precompute the array $M[][]$ for indexing the vertices.
Finally, we build there the new subtree
and then remove superfluous vertices from the intervals.

\begin{theorem}[complexity of subtree rebuild]\label{lem:rebuild_io}
	The I/O complexity of rebuilding a subtree of the target size $S$ is
	$\cO(S/B + k⋅(\log^2 N)/B + l⋅(\log_B N)/ε + ε^{-2})$,
	where $k$ is the sum of changes of number of vertices in each memory interval
	and $l≤k$ is the number of memory intervals where inserts/removals occurred.
	The time complexity is $\cO(S + k⋅\log^2 N)$.
	We can upperbound the complexity by counting each subtree level as an individual memory interval.
\end{theorem}
\begin{proof}
	The subtree build itself has complexity $\cO(S/B + ε^{-2})$ as proven in Lemma~\ref{lem:general_subtree_build}.
	Identifying relevant memory intervals and constructing the array $M[][]$
	has the I/O complexity bounded by the DFS of the subtree and sequential scan of the intervals,
	which is also $\cO(S/B + ε^{-2})$ by Theorem~\ref{thm:dfs}.
	Inserting/removing $k$ vertices in $l$ memory intervals has complexity
	$\cO(k⋅(\log^2 N)/B + l⋅(\log_B N)/ε + 1)$ by Theorem~\ref{thm:pma+dfs}.
	The time complexity follows from the same claims.
\end{proof}

Rebuilding several sibling subtrees works similarly and leads to the same complexity,
where $S$ is the sum of their sizes,
because memory intervals of two sibling subtrees containing the same tree levels are beside each other in memory.

\section{Different Variants and Future Work}
\label{sect:variants}
\subsection{Randomized PMA}

In the course of writing this paper,
Bender et al.~\cite{BCFCKKW22} proposed a new solution to~the~online list labeling problem
improving the number of relabelings per update in linear space from~$\cO(\log^2 N)$ to expected $\cO(\log^{1.5} N)$.
An alternative PMA can be based on~the~result.

The new PMA still divides its memory into logarithmically-sized segments and builds a~binary tree over them.
The main difference from the version described in Section~\ref{sect:pma} is
that during update the tree is traversed in the top-down direction,
i.e. from the root to~the~node to be rebalanced, which contains the insert/delete position.
In each traversed node, it is determined whether the rebalancing occurs there
according to the density of the node, other constant-sized node data and randomness.

In the original PMA, we computed the densities on the fly by scanning the subarrays of~the~traversed nodes,
which had no impact on the asymptotic complexity as we traversed the~tree in the bottom-up direction.
We cannot replicate it here, so we have to store the~densities along with the whole tree explicitly,
which incurs other memory transfers.
If~we store the~tree in the vEB layout described in Section~\ref{sect:veb},
we can perform updates in an expected I/O complexity $\cO((\log^{1.5} N)/B + (\log_B N)/ε + 1)$
instead of the amortized $\cO((\log^2 N)/B + 1)$ of the original PMA.

When inserting or removing a sequence of $S$ consecutive elements,
it is sufficient to~traverse the tree only once and
so the complexity is $\cO(S⋅(\log^{1.5} N)/B + (\log_B N)/ε + 1)$.
The~complexity of inserting or removing a subtree of size $S$ in \corobts{} then decreases to~%
$\cO(S⋅(\log^{1.5} N)/B + (\log_B N⋅\log\log S)/ε+1)$ and become expected instead of amortized.

To use \corobts{} based on the new PMA in the persistent array,
the analysis has to handle the PMA tree in the same way as the stored tree w.r.t. the cache;
i.e. if a consecutively stored subtree $T$ of the top ST-tree has to be kept in the cache,
leaves of the PMA tree of~the~segments intersecting $T$ and their ancestors will be also kept in the cache.
We get the~amortized expected complexity with the appropriately lowered exponent in logarithm.

\subsection{Other Directions}

Bender and Hu~\cite{BH07} developed an \textit{adaptive packed-memory array},
which improves the I/O complexity of~specific patterns of updates from $\cO((\log^2 N)/B)$ to $\cO((\log N)/B)$.
It does not apply directly, but further analysis may get it useful.

In 2006, Bender et al.~\cite{BFCK06} came with an idea of aligning PMA segments to subtrees.
This approach may help avoiding the~$\cO(\log_B N)$ term in \corobts{} I/O complexity,
but it cannot be used straightforwardly due to its requirement of~same-height subtrees having similar size,
which is not fulfilled here.

Other improvements of packed memory array were proposed by Leo and Boncz~\cite{LB19}.

Finally, there is a data structure similar to PMA which permits some segments to be swapped
improving the~I/O complexity of updates to $\cO((\log\log N)^{2+ε})$ \cite{BCDF02}.
Replacing PMA with this data structure may possibly also improve the I/O complexities of updates,
but it will make analysis much more complicated.

As a minor improvement, using the hierarchical memory layouts of Lindstrom and Rajan~\cite{LR14} instead of the vEB layout
may improve search complexity by a~constant factor.

\section{Application of \corobts{} in Persistent Array}
\label{sect:persistent_array}
\paragraph*{The original construction of \cite{DFIO14}}

We think of writes as of points in two dimensional space-time,
where indices of cells in~the~array grows to the right and versions in the up direction;
each write overwrites one cell and creates a~new version.
These points are stored in so called \textit{space-time trees} (or \textit{ST-trees}).
We denote $U$ the size of the array in one version and $V$ the total number of writes.

Each node of an ST-tree represents a space-time rectangle.
It is open in the up direction
if~there is nothing above it;
we can still write into such a rectangle.
Children of a node partition its rectangle.
Roots of ST-trees represent rectangles of width $U$ (whole array)
and height at~most~$U$.
After each $U$ writes the rectangle of the root of \textit{top} ST-tree is closed
and a new ST-tree is created above it.
We call all ST-trees under the top one \textit{bottom}.

A subtree of height $h$ containing leaves represents a rectangle of width $2^{h-1}$;
all leaves are at the same depth.
At the beginning, an ST-tree is binary, has just $U$ leaves, which is a power of~two,
and each of them represents one cell of the array.
When needed, we can add somewhere a third child,
which has the same horizontal boundary as one of the original children and is above it;
it thus stores newer versions of the same subarray.

Each node stores the coordinates of its rectangle.
When we add a new node representing a~rectangle above another one,
we must \textit{close} the bottom rectangles by adjusting their upward boundaries,
which were originally infinite.
Each leaf stores a value of its cell at the time of~the~bottom boundary of its rectangle,
and if a point is contained within the rectangle,
the~leaf stores also a value of that single point.

During persistent read, we just need to navigate to the leaf whose rectangle contains the~given coordinates~--
cell index and version.
The leaf either contains the wanted point or the write occurred somewhere below it
and we return the value at the bottom boundary.
We store a~pointer to the leaf of the last persistent read query
and on next query start navigating from this point to speed-up nearby reads.

We say a leaf is \textit{full} if it contains a point;
an internal node is full if it has two full children.
When we add a point to a leaf,
the leaf and possibly some of its ancestors become full.
We~find the bottommost non-full ancestor and \textit{expand it}.
Expansion means that we add to~the~node a third child
whose rectangle will be above the full one
and initialize its newly created binary subtree.
It involves copying values on its bottom boundary to its leaves.
We also close all the open rectangles under the newly created one,
including the newly full nodes.
After the~operation, all the open rectangles are non-full.

To maintain locality, ST-trees are stored in van Emde Boas memory layout
parametrized by~$0<ε≤1/2$. %
To allow efficient insertions of new subtrees,
the top ST-tree is stored as the complete ternary tree of height $\log U + 1$.
It thus occupies $Θ(U^{\log_2 3})$ of space,
but at most $\cO(U\log U)$ is actually used.
When a new top ST-tree is created, the~previous one is \textit{compressed}
so that it occupies only the space which is needed.

Apart from ST-trees, the data structure contains also an ordinary array with the present version
and pointers from each its cell to the corresponding ST-tree leaves representing the~present version of the cell.
Present reads are then handled directly by the ordinary array
and writes can begin with overwriting the value and following the pointer to the open present leaf
where the new value will be also written.
The last part is the log of all writes which have been performed so far,
which is used for rebuilding the data structure if the array is resized.

\begin{lemma}[Davoodi et al., \cite{DFIO14}]\label{lem:pa_2}
	Any induced binary tree of height $h$ will be stored in $\cO(1 + {2^h/ B^{(1-ε)/ \log 3}})$ blocks.
\end{lemma}
\begin{lemma}[Davoodi et al., \cite{DFIO14}]\label{lem:pa_3}
	The amortized number of times a node gains a third child following an~insertion into its subtree of height $h$ is $\cO(1/2^h)$.
\end{lemma}
\begin{lemma}[Davoodi et al., \cite{DFIO14}]\label{lem:pa_10}
	Compression of the top tree into a bottom tree costs $\cO(U\log U / B^{(1 - ε)/ \log3})$ block transfers.
\end{lemma}

\paragraph*{Our version}

We modify the persistent array by storing the topmost ST-tree in our \corobts{} data structure
with arity between 2 and 3.
This implies that there are no pointers to parents within the~tree
and so it is useless to store pointers to leaves from outside the tree.
We replace~the pointer leading to the leaf of the~last persistent read
by a \textit{finger}, which consists of pointers to all vertices on the branch containing that leaf.
We then use the finger as a stack of ancestors to~walk around the tree.
We also need to avoid pointers from cells of the array of the present version.
Here we use the same technique and maintain the second finger to the leaf of the~last written position.
We need to update the pointers of fingers during pointer recalculation,
which has no impact on the complexity as long as there are only constant number of fingers.

\begin{theorem}[space usage]\label{thm:pers_array_space}
	The space complexity of the persistent array with \corobts{} is
	$\cO(U + V\log U)$.
\end{theorem}
\begin{proof}
	An ST-tree has size $\cO(U)$ on its creation.
	By Lemma~\ref{lem:pa_3}, each point contributes by $\cO(2^h⋅1/2^h)$ of space to expansion of each of $\cO(\log U)$ ancestors of the corresponding leaf,
	where~$h$~is the height of the subtree to be inserted.
	Each point in the top ST-tree also contributes $\cO(1)$ space to the creation of a new ST-tree.
	The space required for all ST-trees is~$\cO(U + V\log U)$.

	The present array has $\cO(U)$ cells and the log of all writes $\cO(V)$ items.
\end{proof}

The complexity of present and persistent reads stays unaffected.

To compuete the I/O complexity of writes,
we first adapt our analysis of the \corobts{} data structure
to make it better fit here.
We need to take into account that some blocks are already present in cache when updating \corobts{}.
In~\cite{DFIO14}, it is also assumed that the cache has at least $d⋅\log_B U$ blocks
for a~sufficiently large constant $d$;
otherwise no cache is available to the ephemeral algorithm during persistent reads.
We will utilize this assumption, too.
The main point of the following technical lemma is
that we can pay only the amortized part of the I/O complexity of subtree insertion
if the subtree is sufficiently small
to be whole preloaded in cache along with other required data.
We also have to keep in cache all the data corresponding to the loaded ephemeral blocks
in spite of PMA movements.

\begin{lemma}[subtree insertion]\label{lem:subtree_insertion}
	The I/O complexity of subtree insertion into \corobts{}
	is~$\cO((S/B)⋅\log^2 U + C⋅(\log_B U)/ε)$,
	where $S$ is the size of the inserted binary subtree
	and $C≥1$ the number of blocks it will occupy,
	under the assumption that at least $d⋅\log_B U$ blocks fit into cache for a sufficiently large constant $d$,
	which depends on $ε$.

	Let us further assume that we have preloaded $\cO((\log_B U)/ε)$ blocks of internal data shared by all PMA updates,
	a memory interval containing the whole subtree of the parent node of the new subtree,
	including space for the new subtree and $Θ(B⋅(\log_B U)/ε)$ extra space at~the~beginning and end of the interval;
	let us also assume that the ancestors of all vertices stored in the interval are also loaded.
	Then the I/O complexity of subtree insertion is $\cO((S/B)⋅\log^2 U)$
	and the preloaded parts will be updated to correspond to movements of~the~vertices
	and creation of the new subtree, which will be also loaded.
\end{lemma}
\begin{proof}
	By Lemma~\ref{lem:pma2},
	changes in PMA cost $\cO((S/B)⋅\log^2 U + C)$ or $\cO((S/B)⋅\log^2 U)$ with preloaded memory,
	both excluding pointer recalculation.

	By Theorem~\ref{thm:dfs}, the complexity of one pass of pointer recalculation is $\cO(|P|/B + (\log_B U)/ε)$,
	where $P$ is the memory interval to be recalculated;
	by Theorem~\ref{thm:pma+dfs}, it sums up in an amortized sense to $\cO(S⋅(\log^2 U)/B + L⋅(\log_B U)/ε)$,
	where $L$ is the number of calls to pointer recalculation; we discarded additive constants here assuming $U>B$.
	For one pass of pointer recalculation it is sufficient to load
	all blocks of the recalculated memory interval, ancestors of all vertices contained there and some other data internally used by PMA,
	which either is contained within $\cO((\log_B U)/ε)$ blocks shared between all passes, or is bounded by $\cO(|P|/B)$.

	Each memory interval for recalculation intersects one of the $C$ blocks in which the new tree will be stored.
	For each pass, we load in cache the intersected block, $Θ(B⋅(\log_B U)/ε)$ vertices around it and the ancestors of all these vertices;
	that all occupy $\cO((\log_B U)/ε)$ blocks by Theorem~\ref{thm:dfs}, so we can keep it in cache between two consecutive passes if they share their intersected block.
	Now, all needed vertices for one pass are either preloaded or the $\cO(|P|/B)$ part dominates the complexity,
	so we can bound the amortized cost of PMA updates including~pointer recalculation by $\cO((S/B)⋅\log^2 U + C⋅(\log_B U)/ε)$.

	For the second part of the lemma, for each DFS pass we have either preloaded everything needed
	or $\cO(|P|/B)$ dominates, so we avoid the $\cO(C⋅\log_B U)$ part completely.

	During the DFS all blocks in which movements occur and all ancestors of their vertices are loaded,
	so we can update which of them stays preloaded after the operation.

	After reserving space for the new subtree, it will be build there,
	which is done by just another DFS not worsening the complexity.

	Before PMA updates, we need to traverse the rightmost branch of the left sibling of~the~newly inserted subtree.
	This costs just another $\cO((\log_B U)/ε)$ or is included in the preloaded memory.

	The other data used internally such as local variables and extendable vectors
	may be kept allocated between calls.
	We either keep it whole in the cache,
	or keep there only its beginnings of size $\cO((\log_B U)/ε)$
	and pay loading of the rest by $\cO(|P|/B)$.\

	The two fingers mentioned earlier can be kept in cache and updated at no cost.
\end{proof}

Following \cite{DFIO14},
we analyze the complexity of writes w.r.t. a given \textit{ephemeral} algorithm running in~the~persistent memory,
which executes a sequence of writes that would require $T(M',B')$ block transfers when executed ephemerally
on an initially empty cache with size $M'$ and block size $B'$.
We lower both the cache and block size of the ephemeral algorithm and with each smaller \textit{ephemeral} block loaded by the algorithm,
we load also the necessary part of the ST-tree and keep it in cache as long as the ephemeral block is kept there.
Assuming these data to be preloaded we calculate the remaining required block transfers.

We express the complexity as the sum of two terms:
The first one represents the memory described in the previous paragraph
and the strategy for loading and evicting these full-sized blocks is derived from the ephemeral algorithm.
The second term represents the remaining blocks needed by persistent array,
which are loaded on demand and usually are not kept in~the~cache for a long time.

\begin{lemma}[ephemeral memory]\label{lem:writes_ephemeral_io}
Let us load for each ephemeral block in memory the~leaves representing the present state of all ephemeral memory cells contained in the block,
whole %
decomposition subtrees of height between $\log_3 B^{1-ε}$ and $\log_3 B$ containing the leaves,
space after the subtrees as if they were ternary
and other $Θ(B⋅(\log_B U)/ε)$ of space before and after them.
Considering all full-sized blocks containing all the previous,
we also load all ancestors of all vertices contained in these blocks.
We note that all the subtrees have the~same height and each leaf is contained in just one of them.

Managing the memory controlled by the ephemeral algorithm then requires
\[
	\cO\left(T\left(\frac{εcM}{B^{1-\frac{1-ε}{\log 3}}\log_B U}, B^\frac{1-ε}{\log 3}\right)(\log_B U)/ε\right)
\]
block transfers during a sequence of writes with ephemeral I/O complexity $T(M',B')$,
where $0<ε≤1/2$ and $c$ is a sufficiently small constant,
using $1/3$ of the real cache.
\end{lemma}
\begin{proof}
	A ternary tree of the specified height fits into $\cO(1)$ blocks of memory
	and contains at~least $Ω(2^{\log_3 B^{1-ε}}) = Ω(B^{(1-ε)/\log 3})$ present leaves.
	One ephemeral block is then stored in~$\cO(1)$ of these ternary trees, which occupy $\cO(1)$ full-sized blocks.

	By Theorem~\ref{thm:dfs}, all vertices stored within a block, $Θ(B⋅(\log_B U)/ε)$ vertices before and after the block
	and all their ancestors occupy $\cO(1 + (\log_B U)/ε) = \cO((\log_B U)/ε)$ blocks assuming $U>B$.
	This is all we need to have loaded for one ephemeral block.

	In order to keep in cache all ephemeral blocks which are supposed to be there, the needed reserved cache size is
	\[
		\frac{εcM}{B^{1-\frac{1-ε}{\log 3}}\log_B U} ⋅
		\frac{1}{B^\frac{1-ε}{\log 3}} ⋅
		\cO((\log_B U)/ε) ⋅
		B
		= M/3
	\]
	for a sufficiently small constant $c$.

	Having stored all these blocks in cache, we can freely move our finger between the~leaves at~no~cost.
	If the ephemeral algorithm evicts a block by another,
	we just replace $\cO((\log_B U)/ε)$ full-sized blocks in cache.
\end{proof}

\begin{lemma}[single write I/O]\label{lem:writes_rest_io}
The amortized I/O complexity of a single write operation is
\[
	\cO\left(\frac{\log^3 U}{εB^\frac{1-ε}{\log 3}}\right)
\]
under the assumption that
the ephemeral block to be written into, all its contemporary parts mentioned in Lemma~\ref{lem:writes_ephemeral_io}
and $\cO((\log_B U)/ε)$ blocks shared by all PMA updates from Lemma~\ref{lem:subtree_insertion} are already loaded in the cache.

After the operation, the set of blocks required to be loaded by Lemma~\ref{lem:writes_ephemeral_io} might have changed,
because present state is represented by different vertices and some vertices might have been moved by PMA.
Loading missing blocks is also included in the complexity of write operation.

\end{lemma}
\begin{proof}
	We first consider the common case,
	where
	the top ST-tree contains less than $U$ points.
	During write operation, one of the ancestors of the present leaf corresponding to~the~cell we are writing to
	has to be expanded, which involves insertion of a new subtree under the~node.
	We calculate the amortized complexity of the write as the sum of its contributions to~all~the~ancestors:
	\[
		\cO\left(\sum_{j=1}^{\log U} \frac{1}{2^j} ⋅ \frac{2^j}{B^\frac{1-ε}{\log 3}} ⋅ (\log^2 U)/ε\right)
		=\cO\left(\frac{\log^3 U}{εB^\frac{1-ε}{\log 3}}\right).
	\]

	The sum is over the heights of subtrees rooted at the nodes for expansion.
	The first fraction in~the~sum divides the cost of the expansion into $2^h$ insertions,
	which has to occur beforehand by Lemma~\ref{lem:pa_3}.
	The rest is the cost of the expansion.
	We distinguish two cases.

	In the first case, the second fraction is at least 1.
	The size of subtree to be inserted is $2^{j-1}$
	and by Lemma~\ref{lem:pa_2} it occupies $\cO(2^{j-1}/B^{(1-ε)/\log 3} + 1) ⊆ \cO(2^j/B^{(1-ε)/\log 3})$ blocks.
	By Lemma~\ref{lem:subtree_insertion},
	the complexity of subtree insertion is
	$\cO((2^{j-1}/B)⋅\log^2 U + (2^j/B^{(1 - ε)/\log 3})⋅(\log_B U)/ε) ⊆
	\cO((2^j/ B^{(1-ε)/\log 3}) ⋅ (\log^2 U)/ε)$.
	We further need to populate the inserted subtree with the~values stored in the corresponding binary subtree
	representing the previous present state of~the~ephe\-meral memory and close rectangles in that subtree.
	The old subtree is stored in the same number of blocks as the new one.

	In the second case, the second fraction is less than 1,
	the subtree of the node to be expanded has height less than $\log_3 B^{1-ε}$
	and so by Lemma~\ref{lem:writes_ephemeral_io} it is whole preloaded in the cache
	along with $Θ(B⋅(\log_B U)/ε)$ cells around it and the corresponding ancestors.
	By Lemma~\ref{lem:subtree_insertion}, we need to load only additional
	$\cO((2^{j-1}/B)⋅\log^2 U) ⊆ \cO((2^j/ B^{(1-ε)/\log 3}) ⋅ \log^2 U)$ blocks.
	Both the old and new trees are part of the preloaded memory, so we can freely walk through them.

	If the top ST-tree already contained $U$ points,
	we need to create new ST-tree.
	As this occurs once per $U$ writes,
	by Lemma~\ref{lem:pa_10} the amortized cost is $\cO((\log U)/B^{(1-ε)/\log 3})$.
\end{proof}

Resizing the array by doubling its size is not analyzed here
as it is not different to~the~original construction.
The final I/O complexity of writes is as follows:

\begin{theorem}[write I/O]\label{thm:pers_array_io}
	Let $T(M',B')$ be the ephemeral I/O complexity of a sequence of~writes
	with an initially empty cache of size $M'$ and block size $B'$.
	Then the total amortized I/O complexity of the sequence of $k$ writes into the persistent array is
	\[
		\cO\left(T\left(\frac{εcM}{B^{1-\frac{1-ε}{\log 3}}\log_B U}, B^\frac{1-ε}{\log 3}\right)(\log_B U)/ε
		+ k⋅\frac{\log^3 U}{εB^\frac{1-ε}{\log 3}}\right)
	\]
	for a sufficiently small constant $c$.
	If each cell is written to at most once during the sequence, the complexity is
	\[
		\cO\left(\frac{1}{ε}⋅T\left(\frac{εcM}{B^{1-\frac{1-ε}{\log 3}}\log U}, \frac{B^\frac{1-ε}{\log 3}}{\log B}\right)\log_B U⋅\log^2 U\right).
	\]
\end{theorem}
\begin{proof}
	The first complexity follows directly from Lemma~\ref{lem:writes_ephemeral_io} and Lemma~\ref{lem:writes_rest_io}.
	We get the~second one by bounding $k$ by the number of preloaded cells with lowered ephemeral block and cache size by another factor of $\log B$:
	\[
		k ∈ \cO\left(T\left(\frac{εcM}{B^{1-\frac{1-ε}{\log 3}}\log U}, \frac{B^\frac{1-ε}{\log 3}}{\log B}\right)⋅\frac{B^\frac{1-ε}{\log 3}}{\log B}\right)
	\]
\end{proof}

\bibliography{bibliography}

\appendix

\section{Algorithms}
\label{sect:appendix_algorithms}

\begin{myalg}{
	Initializing a new instance of the data structure
	by building $a$-ary tree of~height~$H$.
	See also Algorithms~\ref{alg:PMABatchUpdate}, \ref{alg:BuildSubtree} and~\ref{alg:NewSubtreeIntervalsSizes}.
}{alg:Init}
\fn Init (integer $H$)
	$h[] ←$ new array of size $H$
	$h[0]=H$
	$i ← 0$
	for $i ← \{0,\dots,H-1\}$:
		$h' ← h[i]$
		while $h'>1$:
			$h'' ← h'$
			$h'←\max(⌊εh'⌋,1)$
			$h[i+h'] ← h''-h'$
	array $N ←$ \fn NewSubtreeIntervalsSizes ($0$)
	array $M ←$ \fn PMABatchUpdate ((\INS, \nil, $N$))
	$v ←$ \fn BuildSubtree ($0, M$)
	store $v$ as root of the tree
\end{myalg}

\begin{myalg}{
	Inserting a new $a$-ary subtree as $c$-th child of vertex $v$.
	See also Algorithms~\ref{alg:PMABatchUpdate}, \ref{alg:BuildSubtree}, \ref{alg:NewSubtreeIntervalsSizes}, \ref{alg:SubtreeIntervalsBeginnings} and~\ref{alg:SubtreeIntervalsEnds}.
}{alg:InsertSubtree}
\fn InsertSubtree (vertex $v$, integer $c$)
	if $c > 0$:
		array $V ←$ \fn SubtreeIntervalsEnds ($(c-1)$-th child of $v$)
	else: # $c = 0$
		array $V ←$ \fn SubtreeIntervalsBeginnings ($c$-th child of $v$)
		for each $w ∈ V$: $w ←$ the vertex just before $w$ in memory
	array $N ←$ \fn NewSubtreeIntervalsSizes ($d(v) + 1$)
	array $O ←$ array of triples (type = \INS, vertex, count) of same size as $N$
	for $i ∈ \{0, ..., N$.size$ - 1\}: O[i] ←$ (\INS, $V[i]$, $N[i]$)
	array $M ←$ \fn PMABatchUpdate ($O$)
	$w ←$ \fn BuildSubtree ($d(v) + 1, M$)
	add $w$ as $c$-th child of $v$
\end{myalg}

\begin{myalg}{
	Removing subtree rooted at $c$-th child of vertex $v$.
	See also Algorithms~\ref{alg:PMABatchUpdate}, \ref{alg:SubtreeIntervalsBeginnings} and~\ref{alg:SubtreeIntervalsEnds}.
}{alg:RemoveSubtree}
\fn RemoveSubtree (vertex $v$, integer $c$)
	$U ←$ \fn SubtreeIntervalsBeginnings ($c$-th child of $v$)
	$V ←$ \fn SubtreeIntervalsEnds ($c$-th child of $v$)
	remove $c$-th child from $v$
	array $O ←$ array of triples (type = \RM, vertex, vertex) of same size as $U$
	for $i ∈ \{0, ..., U.$size$ - 1\}: O[i] ←$ (\RM, $U[i]$, $V[i]$)
	\fn PMABatchUpdate ($O$)
\end{myalg}

\begin{myalg}{
	Reserving or disposing space in PMA, multiple such operations may be performed by one call.
	We are given an array $O$ of operations to be performed.
	Each operation is a triple of its type (\INS{} or \RM)
	and for \INS{} one existing vertex (or \nil) and number of~vertices to~be inserted after it,
	for \RM{} the first and last vertex of interval to be removed.
	All vertices for removal are supposed to be already disconnected from the tree.
	The method moves some vertices within~$\P$ (or even resizes it)
	and recalculates affected pointers between vertices using \fnn{RecalculatePointers}.
	It returns two-dimensional array $M$ of pointers to new vertices in $\P$;
	$M[i][j]$~points to the $j$-th new vertex after the~vertex of $i$-th \INS{} operation.
	See also Algorithms~\ref{alg:PMAGetIntervals}, \ref{alg:PMACalcNewPositions} and~\ref{alg:RecalculatePointers}.
}{alg:PMABatchUpdate}
\fn PMABatchUpdate (array $O$ of operations)
	# transform pointers to indices in $\P$
	array $O' ← O$ as array of triples ($.t$ = \INS, $.i$, $.n$) and ($.t$ = \RM, $.i$, $.j$)
	# create array $I$ of intervals for vertex redistribution and prepare $\P'$
	$I, m ←$ \fn PMAGetIntervals ($O'$)
	$\P' ← \P$ (as reference) if $m = \P$.size else new array of $m$ empty vertices
	# write new positions into vertices, create array $M$ of references to new vertices
	$M ←$ empty array with same capacity as $O'$ of arrays of vertices in $\P'$
	\fn PMACalcNewPositions ($O', I, \P', M$)
	# create array $B$ of updated children pointers (preallocated, size as $\P$; or extendable)
	$B ←$ empty array of triples (vertex, child index, child vertex)
	for intervals ($l, r, \_$) in $I$:~~~\fn RecalculatePointers ($\P[l..r], B$)
	# apply PMA movements
	for each ($l, r, \_$) in $I$:
		# backward scan: move entries of $\P$ to the end of interval, clear removed
		$j ← r$
		for i in ($r, \dots, l$):
			if $\P[i]$ is non-empty:
				if $P[i]$ contains its new position:
					move $\P[i]$ to $\P[j]$
					$j ← j - 1$
				else:~~~make $\P[i]$ empty
		# forward scan: move entries to their final positions
		while $j < r$:
			$j ← j + 1$
			move $\P[j]$ to the position stored inside it and clear that value
	if $\P ≠ \P'$:
		delete array $\P$
		$\P ← \P'$  (as reference)
	# apply changed pointers from $B$
	for each (vertex $v$, integer $c$, vertex $w$) in $B$:
		set $c$-th child of $v$ to $w$ (in $\P$)
		clear $(v, c, w)$ item in $B$
	# initialize new vertices
	for each $v$ in each array of $M$:~~~initialize $v$
	return $M$
\end{myalg}

\begin{myalg}{
	Part of \fnn{PMABatchUpdate}.
	It takes the array of operations $O'$ with indices to $\P$ instead of pointers;
	as defined in \fnn{PMABatchUpdate}.
	It finds memory intervals whose vertices have to be uniformly redistributed.
	It returns an array $I$ and needed size of $\P$.
	The~array~$I$ consists of triples, each containing indices of first and last vertex of the interval
	and number of vertices which will be in that interval after performing all operations in~$O'$.
	If the needed size of $\P$ is different to the current one, then $I$ contains the only interval referring to whole $\P$.
}{alg:PMAGetIntervals}
\fn PMAGetIntervals (array $O'$)
	$I ←$ empty array of triples (first index $.l$, last $.r$, count of vertices $.n$), capacity as $O'$
	$o ← 0$                            # index in $O'$
	while $o < O'$.size
		# init scan in PMA block where current operation $O'[o]$ occurs
		$m ← ⌈⌈\log(\P.\f{size} + 1)⌉⌉$           # nodes' capacity on current (now leaf) level; power of two
		$\rho_{\min},\;\rho_{\max} ← 1/8,\; 1$        # density bounds on current level, top-level is $(2/8, 7/8)$
		$\rho_\Delta ← 1/(8\log_2(\P.\f{size}/m))$  # difference of $\rho_{\min},\rho_{\max}$ between levels, NaN on init
		if $O'[o] = \INS{}$ and $O'[o].i = -1$: # special case: inserting at the beginning
			$b ← 0$;~~$n ← O'[o].n$;~~$o ← o + 1$
		else:
			$b ← ⌊O'[o].i / m⌋$            # current node number (of capacity m) on the level
			$n ← 0$                             # number of found vertices in the node
		$l,\; r ← b ⋅ m,\; (b+1)⋅m-1$  # first and last index of current node
		$i ← l;~~d ← 1$                         # current index within node and scan direction (+1 or -1)
		repeat:
			# scan in direction $d$
			while $l ≤ i ≤ r$:
				if $\P[i]$ is non-empty:
					$n ← n + 1$
					if $d = 1$ and $o < O'.\f{size}$: # take changes in $O'$ into account in forward scan
						if $O'[o].t = \RM{}$ and $O'[o].i ≤ i ≤ O'[o].j$:
							$n ← n - 1$
							if $i = O'[o].j$:~~~$o ← o + 1$
						if $O'[o].t = \INS{}$ and $O'[o].i = i$:
							$n ← n + O'[o].n$;~~$o ← o + 1$
				$i ← i + d$
			# determine new scan direction, merge with previous interval if possible
			repeat:
				$d ← -1 ^d$
				if $d = 1$ or $I$ is empty or $I.\f{last}.r + 1 ≠ l$:~~~break
				$l ← I.\f{last}.l;~~n ← n + I.\f{last}.n;~~m ← 2m;~~\;b ← ⌊b/2⌋;$~~delete $I.\f{last}$
				$\rho_{\min},\; \rho_{\max} ← \rho_{\min}+\rho_\Delta,\; \rho_{\max}-\rho_\Delta$
			# check density
			if $\rho_{\min} ≤ n / m ≤ \rho_{\max}$:~~~$I$.append($(l, r, n)$);~~break
			# check if total rebuild is needed; get top-level density in (3/8, 6/8) in that case
			if $m = \P.\f{size}$:~~~I.append($(l, r, n)$);~~return $I,\; ⌈⌈N⋅4/3⌉⌉$
			# prepare for scan continuation
			if $d = 1$:~~~$i ← r + 1;~~r ← r + m$
			else:~~~$i ← l -1;~~l ← l - m$
			$m ← 2m;~~b ← ⌊b/2⌋$;~~$\rho_{\min},\; \rho_{\max} ← \rho_{\min}+\rho_\Delta,\; \rho_{\max}-\rho_\Delta$
	return $I,\; \P.\f{size}$
\end{myalg}

\begin{myalg}{
	Part of \fnn{PMABatchUpdate}.
	It computes new positions of vertices after performing all operations.
	It takes the array of operations $O'$ with indices to $\P$ instead of pointers,
	array of~intervals $I$ with resulting number of vertices in them
	and target PMA memory $\P'$, which usually equals $\P$;
	all~as~defined in \fnn{PMABatchUpdate}.
	It writes pointers to new positions of vertices into them
	and writes pointers to target positions of new vertices into $M$,
	which is then returned by~\fnn{PMABatchUpdate}.
}{alg:PMACalcNewPositions}
\fn PMACalcNewPositions (array $O'$, array $I$, PMA memory $\P'$, array $M$)
	$o ← 0$
	for each $(l, r, n)$ in $I$:
		$m ← r - l + 1$
		if $m = \P$.size:~~~$m ← \P'.size$
		$j ← 0$
		if $o < O'$.size and $O'[o] = \INS{}$ and $O'[o].i = -1$:
			# special case: inserting at the beginning
			M.append(empty array with capacity $O'[o].n$ of vertices)
			repeat $O'[o].n$ times:
				$M.\f{last}$.append($\P'[l + j ⋅ m / n]$) (as reference)
				$j ← j + 1$
			$o = o + 1$
		for $i$ in $(l, ..., r)$:
			if $\P[i]$ is non-empty:
				if $o < O'$.size and $O'[o].t = \RM{}$ and $O'[o].i ≤ i ≤ O'[o].j$:
					if $i = O'[o].j$:~~$o ← o + 1$
				else:
					set pointer to new position of $P[i]$ to point to $\P'[l + j ⋅ m / n]$
					$j ← j + 1$
				if $o < O'$.size and $O'[o].t = \INS{}$ and $O'[o].i = i$:
					$M$.append(empty array with capacity $O'[o].n$ of vertices)
					repeat $O'[o].n$ times:
						$M.\f{last}$.append($\P'[l + j ⋅ m / n]$) (as reference)
						$j ← j + 1$
					$o = o + 1$
\end{myalg}

\begin{myalg}{
	Pointer recalculation as part of PMA data movements.
	We are given memory interval $P⊆\P$ of vertices to be moved
	with pointers to new positions within the vertices
	and we are supposed to append to preallocated array $B$ which pointers to children have to be changed and how.
	After returning from this procedure,
	vertices in $P$ will be moved to their new positions
	and then pointers will be updated according to~$B$.
}{alg:RecalculatePointers}
\fn RecalculatePointers (memory interval $P$ of $\P$, ref array $B$):
	# determine the leftmost and rightmost vertex of each level
	$V ←$ array of size $H$ of triples (vertex $.l ← \nil$, vertex $.r ← \nil$, bool $.m ← \fnn{false}$)
	$d_{\max} ← -1$
	for $v∈P$:
		if $V[d(v)].l = \nil$: $V[d(v)].l ← v$
		$V[d(v)].r ← v$
		$d_{\max} = \max(d_{\max}, d(v))$
	# depth-first search
	$v ←$ root of $\T$ # currently visited vertex
	$c ← \fnn{false}$  # whether to continue DFS from $v$
	$S ←$ empty stack of ancestors of $v$
	repeat:
		if $V[d(v)].m = \fnn{true}$:  # $v ∈ P$ and $V ≠$ root of $\T$
			$u ← S.\fnn{top}$ # parent of $v$
			$B.\fnn{append}$((new position of $u$, index of $v$ under $u$, new position of $v$))
		if $c = \fnn{true}$ and $d(v) = d_{\max}$:  # continue DFS by leaving finished levels
			while $v = V[d(v)].r$:  # go up
				$V[d(v)].m ← \fnn{false}$
				$d_{\max} = d_{\max} - 1$
				if $d_{\max} = -1$: return
				$v ← S.\fnn{pop}()$
			if $V[d(v)].m = \fnn{false}$: $c = \fnn{false}$  # no other vertices to visit nearby
		if $c = \fnn{true}$:  # continue DFS to the right
			$v' ←$ leftmost child of $v$, or $\nil$
			while $d(v) = d_{\max}$ or $v' = \nil$:  # go up
				$v' ←$ child of $S.\fnn{top}$ which is right sibling of $v$, or $\nil$
				$v ← S.\fnn{pop}()$
			# step down
			$S.\fnn{push}(v)$
			$v ← v'$
		else:  # navigate to the next leftmost unvisited vertex
			while $V[d_{\max}].l = \nil$:  # identify the unvisited vertex
				$d_{\max} = d_{\max} - 1$
				if $d_{\max} = -1$: return
			while $V[d_{\max}].l$ is not in subtree of $v$: $v←S.$\fnn{pop()}  # go up
			# step down
			$S.\fnn{push}(v)$
			$v ←$ child of v on branch with $V[d_{\max}].l$
			if $v = V[d(v)].l$: $V[d(v)].m = \fnn{true}$
			if $d(v) = d_{\max}$: $c=\fnn{true}$  # the vertex found
\end{myalg}

\begin{myalg}{
Building a new subtree with root at depth $d_0$ in a reserved space given by~$M$.
The~vertex~$M[i][j]$ is the $j$-th one in the $i$-th memory interval.
The subtree is being built via DFS and its root is returned to be connected to its new parent.
}{alg:BuildSubtree}
\fn BuildSubtree (integer $d_0$, array $M[][]$ of vertices in $\P$)
	array $I[d_0..H-1], J[d_0..H-1]$  # positions in $M$ for individual levels
	array $K[d_0..H-1]$               # number of times vertices on last branch were visited
	# initialize $I, J, K$
	$i ← 0$
	$d ← d_0$
	while $d < H$:
		$J[d] ← 0$
		repeat $h[d]$ times:
			$I[d] ← i$
			$d    ← d + 1$
		$i ← i + 1$
	$K[d_0] ← 0$
	# DFS building the new subtree
	$d ← d_0$
	while $d ≥ d_0$:
		$K[d] ← K[d] + 1$
		if $K[d] = 1$:
			# first visit, init vertex data
			$d(M[I[d]][J[d]]) ← d$
		if $d = H - 1$:
			# leaf, go up
			$d ← d - 1$
		else if $K[d] = 1$:
			# first visit, non-leaf, init some vertices on the leftmost branch (green arrows)
			$h' ← h[d]$
			while $h' > 1$:
				$h' ←$ $\max(⌊εh'⌋,1)$
				$J[d + h'] ⇐ J[d] + (a^{h'} - 1) / (a - 1)$  # = $J[d] + \sum_{i=0}^{h'-1} a^i$
			$K[d + 1] ← 0$
			add vertex $M[I[d + 1]][J[d + 1]]$ as first child of vertex $M[I[d]][J[d]]$
			$d    ← d + 1$
		else if $K[d] ≤ a$:
			# other but not last visit, init next child (red arrows)
			$J[d + 1] ← J[d + h[d + 1]] + 1$
			$K[d + 1] ← 0$
			add vertex $M[I[d + 1]][J[d + 1]]$ as next child of vertex $M[I[d]][J[d]]$
			$d    ← d + 1$
		else if $K[d] = a + 1$:
			# last visit, go up
			$d ← d - 1$
	return $M[0][0]$
\end{myalg}

\begin{myalg}{
	Calculating sizes of (possibly isolated) memory intervals needed to store
	a~new $a$-ary (sub)tree with root at depth $d$.
	It returns an array of number of vertices in each memory interval.
	The array have the same indexing as arrays
	from \fnn{SubtreeIntervalsBeginnings} and \fnn{SubtreeIntervalsEnds},
	which allows reserving space for a new sibling subtree of an existing one before or after its intervals.
}{alg:NewSubtreeIntervalsSizes}
\fn NewSubtreeIntervalsSizes (integer $d$)
	array $N[]$
	$n ← 1$  # number of subtrees in i-th memory interval
	while $d < H$:
		$N.\fnn{append}(n ⋅ (a^{h[d]} - 1) / (a - 1))$
		$n    ← n ⋅  a^{h[d]}$
		$d    ← d + h[d]$
	return $N$
\end{myalg}

\begin{myalg}{
	Calculating beginnings of (possibly isolated) memory intervals containing
	an~existing subtree rooted at vertex $v$.
	It returns an array of first vertices in each interval.
}{alg:SubtreeIntervalsBeginnings}
\fn SubtreeIntervalsBeginnings (vertex $v$)
	array $V[]$
	repeat:
		$V.\fnn{append}(v)$
		repeat $(h(v)-1)$ times:
			$v ←$ leftmost child of $v$
		break if $v$ is leaf
		$v ←$ leftmost child of $v$
	return $V$
\end{myalg}

\begin{myalg}{
	Calculating ends of (possibly isolated) memory intervals containing
	an~existing subtree rooted at vertex $v$.
	It returns an array of last vertices in each interval.
}{alg:SubtreeIntervalsEnds}
\fn SubtreeIntervalsEnds (vertex $v$)
	array $V[]$
	repeat:
		repeat $(h(v)-1)$ times:
			$v ←$ rightmost child of $v$
		$V.\fnn{append}(v)$
		break if $v$ is leaf
		$v ←$ rightmost child of $v$
	return $V$
\end{myalg}

\clearpage %

\end{document}